\documentclass[letterpaper,11pt]{article}
\usepackage[usenames,dvipsnames,svgnames,table]{xcolor}
\usepackage[utf8]{inputenc}
\usepackage{amsmath}
\usepackage{amssymb}
\usepackage{amsfonts}
\usepackage{amsthm}
\usepackage{graphicx}
\usepackage{color}
\usepackage{hyperref}
\usepackage{cite}
\usepackage{tikz}
\usepackage{xspace}
\usepackage{needspace}
\usepackage{url}
\usepackage{subcaption}
\usepackage{enumitem}
\usepackage[export]{adjustbox}%
\usepackage[top=1.3in, bottom=1.2in, left=1.1in, right=1.1in]{geometry}
\usetikzlibrary{decorations.pathmorphing}
\usetikzlibrary{shapes}
\newtheorem{theorem}{Theorem}[section]

\newtheorem{lemma}[theorem]{Lemma}
\newtheorem{corollary}[theorem]{Corollary}

\theoremstyle{definition}
\newtheorem{definition}[theorem]{Definition}

\newtheorem{algorithm}[theorem]{Algorithm}
\let\oldalgorithm\algorithm
\renewcommand{\algorithm}{\oldalgorithm\normalfont}

\newcommand{\termasm}[1]{\mathcal{A}_{\Box}[{#1}]}
\newcommand{\prodasm}[1]{\mathcal{A}[{#1}]}

\newcommand{\dom}[1]{{\rm dom}(#1)}

\newcommand{\glue}{\mathcal{G}}

\newcommand{\seed}{\sigma}%

\newcommand{\rstake}{\mathcal{R}^s}
\newcommand{\rdanger}{\mathcal{R}^d}
\newcommand{\cbound}{\mathcal{B}^c}
\newcommand{\fbound}{\mathcal{B}^f}
\newcommand{\bbound}{\mathcal{B}^b}
\newcommand{\ebound}{\mathcal{B}^e}
\newcommand{\dbound}{\mathcal{B}^d}
\newcommand{\seedbound}{\mathcal{B}^\seed}
\newcommand{\sbound}{\mathcal{B}^s}

\newcommand{\fband}{f^b}

\newcommand{\vect}[1]{{\protect\overrightarrow{#1}}}
\newcommand{\resp}{respectively\xspace}
\usepackage{ titling}

\newcommand{\ignore}[1]{}

\renewcommand{\d}[1]{{#1}^\ast}
\newcommand{\zs}{\mathbb{Z}^2}

\title{A pumping lemma for non-cooperative self-assembly}
\author{Pierre-Étienne Meunier\\
  Aalto University \\
\href{mailto:pierre-etienne.meunier@aalto.fi}{pierre-etienne.meunier@aalto.fi}
\thanks{Aalto University, Helsinki, Finland and Aix Marseille Université, CNRS, LIF UMR 7279, 13288, Marseille, France. Supported in part by National Science Foundation Grant CCF-1219274.}
\and
Damien Regnault\\
Université d'Évry Val-d'Essonne\\
\href{mailto:damien.regnault@ibisc.fr}{damien.regnault@ibisc.fr}
\thanks{Université d'Évry Val-d'Essonne, IBISC  EA 4526, 91037, Évry, France. Supported in part by ANR project Quasicool (ANR-12-JS02-011-01)}}

\date{}
\begin{document}

\maketitle

\let\oldlabel=\label
\let\oldref=\ref
\renewcommand\label[1]{\oldlabel{ext-#1}}
\renewcommand\ref[1]{\oldref{ext-#1}}
\let\oldcite=\cite
\renewcommand\cite[1]{\oldcite{ext-#1}}

\newcommand\g{}%

\newcommand\var[1]{\ocwlowerid{#1}}
\renewcommand\label[1]{\oldlabel{#1}}
\renewcommand\ref[1]{\oldref{#1}}
\renewcommand\cite[1]{\oldcite{#1}}
\newcommand\prog{\href{http://users.ics.aalto.fi/meunier/pumpability.html}{\tt http://users.ics.aalto.fi/meunier/pumpability.html}}
\begin{abstract}
We prove the computational weakness of a model of tile assembly that has so far resisted many attempts of formal analysis or positive constructions. Specifically, we prove that, in Winfree's abstract Tile Assembly Model, when restricted to use only noncooperative bindings, any long enough path that can grow in all terminal assemblies is \emph{pumpable}, meaning that this path can be extended into an infinite, ultimately periodic path.

This result can be seen as a geometric generalization of the pumping lemma of finite state automata, and closes the question of what can be computed deterministically in this model.  Moreover, this question has motivated the development of a new method called \emph{visible glues}. We believe that this method can also be used to tackle other long-standing problems in computational geometry, in relation for instance with self-avoiding paths.

Tile assembly (including non-cooperative tile assembly) was originally introduced by Winfree and Rothemund in STOC 2000 to understand how to \emph{program shapes}.
The \emph{non-cooperative} variant, also known as \emph{temperature 1 tile assembly}, is the model where tiles are allowed to bind as soon as they match on one side, whereas in cooperative tile assembly, some tiles need to match on several sides in order to bind.
In this work, we prove that only very simple shapes can indeed be programmed, whereas exactly one known result (SODA 2014) showed a restriction on the assemblies general non-cooperative self-assembly could achieve, without any implication on its computational expressiveness.
With non-square tiles (like polyominos, SODA 2015), other recent works have shown that the model quickly becomes computationally powerful.

\end{abstract}
\clearpage
\section{Introduction}

A possible approach to natural sciences is to try and write programs using the same kind of programming language as we think nature uses. If we can implement our theoretical algorithms in the actual natural systems, we will know the theory is meaningful to the systems studied. Through this process, we can learn from theorems \emph{reasons why} something is true, yielding insights beyond the \emph{modeling} of observed phenomena. This approach contrasts with other approaches where natural scientists test hypotheses against experiments to understand \emph{what happens}.

Although present since Turing~\cite{Turing1936}'s and Von Neumann's~\cite{vonNeumann1966} works, this idea has really been able to develop and extend into physical realizations only in recent years.
One of these realizations is the first implementation of Shor's algorithm in 2001~\cite{vandersypen2001}, providing a precious link between techniques for programming qubits devised by computer scientists, and the bricks actually used by nature.
Another achievement is the connection observed by Winfree in 1998~\cite{Winf98} between core concepts from theoretical computer science (computing machines, tilings) and a kind of building bricks devised by Seeman~\cite{Seem82} using DNA. One of the main models used in this connection, called the \emph{abstract Tile Assembly Model}, has yielded an impressive number of experimental demonstrations~\cite{WinLiuWenSee98,yurke2000dna,RothOrigami}.

Although they use different concepts and tools, these works use the same approach: trying to write programs using the language of nature (of physics in the former case, of chemistry in the latter), and confront these programs to the physical world by implementing them.

In this work, we close a long-standing open problem from the second approach, by proving that a simple version of the programming language of tile assembly, although ubiquitous in many systems, is not sufficient to perform general purpose computation. More precisely, in the abstract Tile Assembly Model, we are interested in the interactions and bindings of grounds of matter represented by square tiles, with glues of a certain \emph{color} and integer \emph{strength} on each of their four borders.
The dynamics start from an initial assembly called the \emph{seed}, and proceeds asynchronously and nondeterministically, one tile at a time, according to the following rule: a tile may attach to the current assembly if the sum of glue strengths on its sides that match the colors of adjacent tiles sum up to at least a parameter of the model called the temperature $\tau=1,2,3\ldots$. In particular, this means that unlike in Wang tilings (one inspiration of this model), adjacent tiles may have a \emph{mismatch}, i.e. disagree on the glue types of their common border.

This model is an abstraction of a simple chemical fact: when the temperature of a solution is increased, so is molecular agitation; for a tile to stay stably attached to an assembly, it needs then either stronger bonds to that assembly, or bonds to a larger neighborhood.

\paragraph{Temperature 1} This work will exclusively focus on \emph{temperature 1} tile assembly, also called \emph{non-cooperative} self-assembly. At higher temperatures, fewer assemblies are stable, allowing more control over producible assemblies: indeed, temperature~2 self-assembly is able to simulate arbitrary Turing machines \cite{Winf98,RotWin00,jCCSA}, and produce arbitrary connected shapes with a number of tile types within a log factor of their Kolmogorov complexity \cite{SolWin07}.  More surprisingly, this model has even been shown intrinsically universal \cite{IUSA}, meaning that there is a single tileset capable of simulating arbitrary tile assembly systems, modulo rescaling. In generalizations of this model, a single tile can even be sufficient to simulate all tile assembly systems and therefore all Turing machines~\cite{Demaine2014}.

In all known generalizations of the model, non-cooperative self-assembly is capable of arbitrary Turing computation: 3D cubic tiles~\cite{Cook-2011}, stochastic assembly sequences~\cite{Cook-2011}, hierarchical self-assembly~\cite{Versus}, polyominoes~\cite{Fekete2014}, duples~\cite{Hendricks-2014}, tiles with signals~\cite{Jonoska2014}, geometric tiles~\cite{geotiles} or negative glues~\cite{Patitz-2011}. Moreover, the synchronous version of this model can simulate arbitrary cellular automata.

The framework of intrinsic simulations (i.e. simulations up to rescaling) has recently yielded the first proof of a qualitative (and indeed geometric) difference between non-cooperative tile assembly and the more general model~\cite{Meunier-2014}. However, that result had no computational implications: indeed, it also holds in the three-dimensional generalization of temperature~1, known to be Turing-universal~\cite{Cook-2011}.

Therefore, an absolute requirement to understand this model seems to be an intuition on the role of planarity, and the shape of tiles.
Here, we introduce a new framework to study how information can be communicated in a planar space, via geometric interactions.
This framework will then (in the end of our proof) allow us to abstract geometric considerations away and reason on large boxes in which paths are forced to grow. This is a significant progress in the field, since the ``low-level geometry'' of paths producible at temperature~1 has been notoriously difficult to understand.

\paragraph{Relation to other works}
As a corollary of our main result, we prove a conjecture by Doty, Patitz and Summers~\cite{Doty-2011}: \emph{for all temperature 1 directed tile assembly systems, there is a constant $c$ such that all paths longer than $c$ producible by the system are pumpable}. By one of their results (conditioned upon the conjecture we are proving here), the set of producible assemblies of temperature~1 directed tile assembly systems is therefore semi-linear, and hence computationally simple.

Moreover, counting and sampling self-avoiding walks in the plane is an old problem at the intersection polymer chemistry and computer science, introduced by Flory~\cite{Flory53}; an early attempt to solve it was made by Knuth~\cite{knuth:math}, and this field has remained active today~\cite{Bousquet2010}. Another interpretation of our questions is the following problem, related to these works: starting from any point in $\zs$, start a self-avoiding walk deterministically (with memory). How far can you go without ever entering a cycle, if you only have $n$ bits of memory?

\subsection{Main result}

Our result can be seen as a two-dimensional equivalent of the pumping lemma on deterministic finite automata \cite{Sipser}: we prove that if a non-cooperative tile assembly system can \emph{always} grow assemblies over a certain size (depending only on the size of their seed and on the number of tile types used), then these paths can be extended into ultimately periodic paths.

However, remark that non-cooperative systems can grow at least the same assemblies as cooperative ones: intuitively, their growth is ``harder to control'', resulting in more possible assemblies.  This is why our result is specific to patterns that can grow \emph{in all assemblies} producible by the system:

\begin{theorem}
\label{def:main:graph}
  Let $\mathcal{T}=(T,\sigma,1)$ be a tile assembly system such that the seed assembly $\sigma$ is finite and connected.  There is a constant $c(|T|,|\dom\sigma|)$ such that any path $P$, that can grow in all assemblies of $\mathcal{T}$ and reaches a point at a distance more than $c(|T|,\sigma)$ from $\sigma$, is \emph{pumpable}.
\end{theorem}

The term \emph{pumpable} will be defined in Section \ref{def:main}. Intuitively, we say that $P$ is pumpable if one of its subpaths $P_{i,i+1,\ldots,j}$ can be repeated infinitely many times immediately after $P_{1,2,\ldots,j-1}$, without \emph{conflicting} with $\sigma$ nor with $P_{1,2,\ldots,j-1]}$, nor with any other repetition.

This theorem implies the longstanding conjecture~\cite{RotWin00, Roth01, Cook-2011,  Doty-2011, Reif-2012, Manuch-2010, Patitz-2011,Meunier-2014, Fekete2014} that no Turing computation can be done in this model in a deterministic way: indeed, an algorithm can characterize producible assemblies by first growing the initial assembly of a constant radius (depending only on the size of the tileset and size of the seed) around the seed, and from there start all possible paths that can go past this radius (there is a finite number of them). By our main theorem, they become periodic after this radius, and hence can be completely characterized algorithmically.

\section{Definitions and basic properties}
\label{def:main}

We first define non-cooperative tile assembly, using standard formalism.

Let $G$ be the grid graph of $\mathbb{Z}^2$, i.e. the undirected graph whose vertices are the points of $\mathbb{Z}^2$, and for any two points $A,B\in\mathbb{Z}^2$, there is an edge between $A$ and $B$ if and only if $\max(|x_B-x_A|,|y_B-y_A|)=1$.%
The \emph{Manhattan distance} between two vertices $u$ and $v$, also written as $||\vect{uv}||_1$, is the length of the shortest path in $G$ between $u$ and $v$. The \emph{Manhattan diameter} of a connected subgraph $G'$ of $G$ is the maximal Manhattan distance of two vertices in $G'$.

Moreover, for any two integers $i$ and $j$ such that $i<j$, we write $[i,j]=\{i,i+1,\ldots,j\}$. When used in the subscript of a sequence, as in $P_{[i,j]}$, this notation means ``the subsequence of $P$ between indices $i$ and $j$ (inclusive)'', or more precisely the sequence $P_i,P_{i+1},\ldots,P_j$.

Finally, paths and sequences in our construction are indexed from 1: the last tile of a path $P$ is written $P_{|P|}$, where $|P|$ means ``the length of $P$''.

\subsection{Cutting the plane with lines}
Sometimes in our proof, we will define cuts of the grid graph of $G$ using lines, which typically live in $\mathbb{R}^2$ and seem to require a planar embedding of the grid graph, as well as ``non-integer half-square-tiles''.

Although our drawing use a planar embedding, our proofs do not. Instead, let $A\in\mathbb{Z}^2$ and let $\vect{v}$ be a vector of $\mathbb{Z}^2$: the \emph{cut of $\mathbb{Z}^2$ by the line of vector $\vect{v}$ passing through $A$} is the following two connected components of $G$:
\begin{eqnarray*}
  U&=&\{X\in\mathbb{Z}^2 | \det(\vect{v},\vect{AX}) \geq 0\}\\
  V&=&\{X\in\mathbb{Z}^2 | \det(\vect{v},\vect{AX}) < 0\}
\end{eqnarray*}

Remark that, in this case, all determinants are integers.

\subsection{The abstract Tile Assembly Model}

A \emph{tile type} is a unit square with four sides, each consisting of a \emph{glue label} and a nonnegative integer strength. Formally, a tile $t=(n,e,s,o)$ is an element  of $(\glue\times\mathbb{N})^4$, where $\glue$ is a finite set of glue labels. Moreover, $n$ is called its \emph{north glue}, $e$ its \emph{east glue}, $s$ its \emph{south glue} and $w$ its \emph{west glue}.

Let $T$ be a finite set of tile types.
An \emph{assembly over $T$} is a partial function of $\mathbb{Z}^2\dashrightarrow T$, whose domain is a connected component of $\mathbb{Z}^2$. Intuitively, an assembly is a positioning of tile types at some positions in the plane.

We say that two neighboring tiles of an assembly \emph{interact} if the glue labels on their abutting side are equal, and have positive strength. An assembly $\alpha$ induces a weighted \emph{binding graph} $G_\alpha=(V_\alpha,E_\alpha)$, where $V_\alpha=\dom\alpha$, and there is an edge $(a,b)\in E$ if and only if $a$ and $b$ interact, with weight the glue strength between $a$ and $b$.
An assembly $\alpha$ is said to be \emph{$\tau$-stable} if any cut of $G_\alpha$ has weight at least $\tau$.

A \emph{tile assembly system} is a triple $\mathcal{T}=(T,\sigma,\tau)$, where $T$ is a finite tile set, $\sigma$ is a $\tau$-stable assembly called the seed, and $\tau\in\mathbb{N}$ is the temperature.
In this paper, $\tau$ will always be equal to~$1$.

Given two $\tau$-stable assemblies $\alpha$ and $\beta$, we say that $\alpha$ is a \emph{subassembly} of $\beta$, and write $\alpha\sqsubseteq\beta$ if $\dom\alpha\subseteq\dom\beta$, and for all position $p\in\dom\alpha$, $\alpha(p)=\beta(p)$. We also write $\alpha\rightarrow_1^{\mathcal T}\beta$ if $\alpha\sqsubseteq\beta$ and $|\dom\beta\setminus\dom\alpha|=1$ (i.e. if we can obtain $\beta$ from $\alpha$ by a single tile attachment).

We say that $\beta$ is \emph{producible} from $\alpha$, and write $\alpha\rightarrow^{\mathcal T}\beta$ (or simply $\alpha\rightarrow\beta$ if there is no ambiguity), if there is a (possibly empty) sequence $\alpha=\alpha_0\rightarrow_1^{\mathcal T}\alpha_1\rightarrow_1^{\mathcal T}\ldots\rightarrow_1^{\mathcal T}\alpha_{n-1}=\beta$.
The set of \emph{productions} of a tile assembly system $\mathcal{T}=(T,\sigma,\tau)$ is $\prodasm {\mathcal T}=\{\alpha | \sigma\rightarrow^{\mathcal T}\alpha\}$. Moreover, an assembly $\alpha$ is called \emph{terminal} if there is $\beta$ such that $\alpha\rightarrow_1^{\mathcal{T}}\beta$, and the set of productions of a tile assembly system $\mathcal T$, that are terminal assemblies, is written $\termasm{\mathcal T}$.

\subsection{Concurrency and conflicts}

The main arguments of this proof take advantage of conflicting assemblies to build an assembly that blocks a path. This involves a subtle technical difficulty, that can be easily dealt with using proper vocabulary. When assemblies overlap, two different things can happen: either the assemblies disagree on the tile types they place at their common positions, or they agree.

In particular, in our lemmas about U-turns (Section~\ref{sec:algo}), one might be tempted to ignore the case where assemblies overlap and agree. Similarly, in some definitions of pumping for a path, the result of the pumping has sometimes been assumed to also be a path, whereas consecutive iterations might simply intersect and agree.

If two assemblies $\alpha$ and $\beta$ agree on all their common positions, i.e. if there is an assembly $\gamma$ of domain $\dom\alpha\cup\dom\beta$ such that $\alpha\sqsubseteq\gamma$ and $\beta\sqsubseteq\gamma$, we write $\alpha\cup\beta$ for this assembly $\gamma$.

If two assemblies $\alpha$ and $\beta$ are such that there is a position $p\in\dom\alpha\cap\dom\beta$, but $\alpha(p)\neq \beta(p)$, we say that $\alpha$ and $\beta$ \emph{conflict} at position $p$. On the other hand, on a position $p\in\dom\alpha\cap\dom\beta$, we say that $\alpha$ and $\beta$ \emph{intersect without conflicting} (or simply \emph{intersect}) if $\alpha(p)=\beta(p)$.

Finally, we say that a tile assembly system is \emph{deterministic} if it has exactly one (potentially infinite) terminal assembly.

\subsection{Paths and path assemblies}
\newcommand\pos{\mathrm{pos}}
\newcommand\type{\mathrm{type}}
An important point about temperature 1 tile assembly, is that any path in the binding graph of an assembly can start to grow, independent from anything else. More precisely, if $\mathcal{T}=(T,\sigma,1)$ is a tile assembly system, then for any $\alpha\in\prodasm{\mathcal{T}}$, and any path $P$ in the binding graph of $\alpha$ such that $P_0$ is in $\sigma$, an immediate induction on the length of $P$ shows that the restriction of $\alpha$ to $\sigma\cup P$ is in $\prodasm{\mathcal{T}}$.

Since assemblies following paths are particularly important in our proof, we define them now using \emph{sequences} instead of the more general formalism of \emph{assemblies}: first, for any element $a=(p,t)\in\mathbb{Z}^2\times T$, we call $p$ the \emph{position} of $a$, written $\pos(a)$, and $t$ the \emph{type} of $a$, written as $\type(a)$. We also write the position of $a$ as $(x_a,y_a)$.
Let then $P$ be any sequence of $\mathbb{Z}^2\times T$. If for all $i,j$, $\pos(P_i)=\pos(P_j) \Rightarrow\type(P_i)=\type(P_j)$ ($P$ might not be simple, but this condition means that any two tiles of $P$ at the same position have the same type), we define the assembly \emph{induced} by $P$ as the assembly  $\alpha_P$ such that $\dom{\alpha_P}=\{\pos(P_i) | i\in\{1,2,\ldots,|P|-1\}\}$ and for all $i$, $\alpha(\pos(P_i))=\type(P_i)$.

Moreover, if the sequence of positions of $P$, i.e. $(\pos(P_i))_{i\in\{1,2,\ldots,|P|\}}$, is a path of $\mathbb{Z}^2$, and if $P$ induces an assembly $\alpha$ such that for all $i\in\{1,2,\ldots,|P|-1\}$, the tiles at positions $(x_i,y_i)$ and $(x_{i+1},y_{i+1})$ in $\alpha$ interact, we call $P$ a \emph{path assembly} (even though not formally an assembly, since a path assembly is a sequence of $\mathbb{Z}^2\times T$, and an assembly is a function of $\mathbb{Z}^2\rightarrow T$).

\subsection{Cutting the plane with paths and lines}

In this proof, we will use quite extensively cuts of $\mathbb{Z}^2$ delimited by path assemblies and horizontal rays. This is not formally correct, since ``horizontal rays'' are formally sequences of edges of the grid graph of $\mathbb{Z}^2$, and path assemblies are sequences of vertices.

To solve this issue, we adopt the following convention: every time a cut is defined in this way in our proof, it will be defined using a sequence of paths, either in the grid graph of $\mathbb{Z}^2$ (defined by path assemblies), or in the dual of the grid graph of $\mathbb{Z}^2$ (defined by rays).

We first disconnect the grid graph of $\mathbb{Z}^2$ by removing all edges from the rays, and all edges connected to the vertices. This produces at least one infinite connected component $A$, and a number of other connected components (in this paper, at most one of these other components is infinite). We then let $B$ be the union of all those connected components. The formal cut of $\mathbb{Z}^2$ defined by this construction is then $(A,B)$.

\subsection{Pumpable paths}

We are now going to define a particular kind of path assemblies, made by repeating a part of a path assembly periodically.
Let $\mathcal{T}=(T,\sigma,1)$ be a temperature 1 tile assembly system, and $P$ be a path assembly producible by $\mathcal{T}$, of length at least 2, and $i,j\in\{1,2,\ldots,|P|-1\}$ two integers such that $i< j$.

The \emph{pumping of $P$ between $i$ and $j$} (also written as ``the pumping of $P_{[u,v]}$'' is the sequence $Q$ of $\mathbb{Z}^2\times T$ defined for all integer $k\in\mathbb{N}$ by $Q_k=P_k$ if $k<i$, and $Q_k=P_{i+((k-i)\mod (j-i))} + \left\lfloor\frac{k-i}{j-i}\right\rfloor\vect{P_iP_j}$.

This definition does not imply that $\sigma\cup Q$ is a producible assembly. However, when $P_i$ and $P_j$ are of the same type, and $\sigma\cup Q$ is not producible, then some tile of $Q$ must be at the same position as another tile of $P_{[1,j]}$.

A path assembly is said to be \emph{pumpable} when one of its pumpings induces an assembly producible by $\mathcal{T}$. Remark that the pumping of a path is a path assembly (i.e. its positions follow a single path in the binding graph of some producible assembly), although not necessarily a \emph{simple} path assembly.

Also, by this definition, any path with a pumpable prefix is also pumpable.

\subsection{Fragility}
\label{def:fragile}

  Let $\mathcal{T}=(T,\sigma,1)$ be a tile assembly system. We say that a path $P$ is \emph{fragile} when there is at least one terminal assembly $\alpha\in\termasm{\mathcal{T}}$ of which $P$ is not a subassembly. Or equivalently, if the assembly induced by $P$ conflicts with $\alpha$.

According to this definition, if an assembly admits a non-fragile path assembly, then it can always be produced from any assembly. It also implies that if different producible assemblies of a tile assembly system with the same domain are all fragile.

\subsection{Visibility}

One of the key insights of this proof is the notion of \emph{visibility}, which allows us to reason about the side of a path enclosed by potential collisions between paths:

\begin{definition}
\label{def:visible}
Let $P$ be a path assembly. For any $i\in\{1,2,\ldots,|P|-1\}$, we say that the glue on the output side of $P_i$ is \emph{visible from the east} if it is to the left of both:
\begin{itemize}
\item all other glues between two consecutive tiles in $P$ that are on the same row as the edge between $P_i$ and $P_{i+1}$.
\item all edges of the grid graph of $\mathbb{Z}^2$, that are between two adjacent tiles of $\sigma$ and on the same row as the edge between $P_i$ and $P_{i+1}$.
\end{itemize}
Moreover, we call the \emph{east-visibility ray of glue $(P_i,P_{i+1})$} (or, when it is clear from the context, the \emph{visibility ray} of that glue) a horizontal ray to the east from that glue.
\end{definition}

Remark that this definition is not completely intuitive: in particular, in the case where two adjacent tiles of $P$ interact, are the leftmost of their rows, but are not consecutive in $P$, the glue between them is not visible, and another glue is visible on that row, as shown on Figure~\ref{fig:opposite}.

\begin{figure}[ht]
  \centering
  \includegraphics{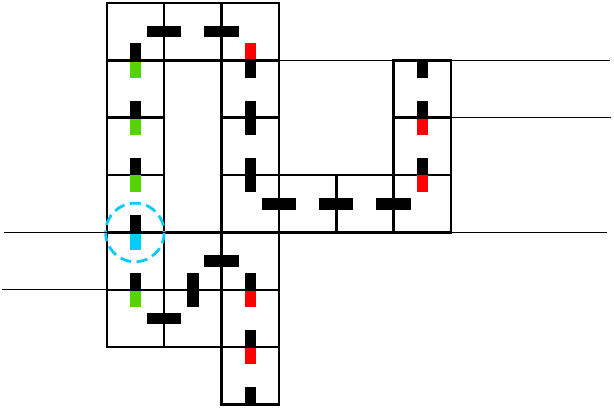}
  \caption{On this drawing, all output glues of tiles, that are visible from the west or the east are colored: in red, output glues that are visible from the west only; in green, output glues that are visible from the east only; in blue (and circled), one output glue is visible from both. the horizontal lines drawn are rays in $\d\zs$, defined in Definition~\ref{def:visible}.
    Here, all output glues but one (the top red glue) are north glues; the characterization of when this happens will be given in Lemma~\ref{lem:watershed0}}
  \label{fig:opposite}
\end{figure}

Our first lemma about visible glues, Lemma~\ref{lem:watershed0}, is exemplified on Figure~\ref{fig:visabove}.
  \begin{figure}[ht]
    \begin{center}
      \includegraphics{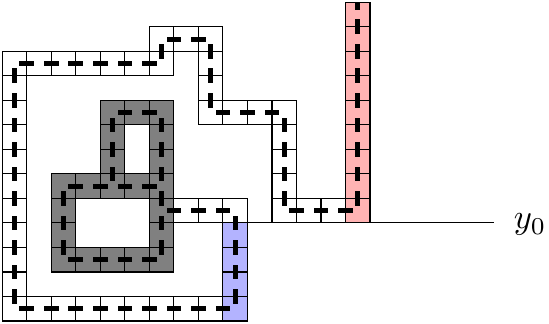}
    \end{center}
    \caption{The visible glue on tiles strictly above $l$ is the north one, and
      below or on $l$ is the south one. On this figure, the seed is in
      gray. Tiles whose north glue is visible from a ray to the East are in red, and
      tiles whose south glue is visible from a ray to the East are in blue.}
    \label{fig:visabove}
  \end{figure}

\begin{lemma}
  \label{lem:watershed0}
  Let $\mathcal{T}=(T,\sigma,1)$ be a tile assembly system such that $\sigma$ is finite and connected, and let $P$ be a path producible by $\mathcal{T}$, whose last point is a highest point of $P$.

  Then there is an integer $y_0$, such that:
  \begin{itemize}
  \item for all $y$ such that $y_{0}\leq y$, all glues of $P$ between rows $y$ and $y+1$ that are visible from the east (if any) are north glues.
  \item for all $y<y_{0}$, all glues of $P$ between rows $y$ and $y+1$ that are visible from the east (if any) are south glues.
  \end{itemize}
\end{lemma}
\begin{proof}
  Let $y_1=y_{P_{|P|}}$. We consider the set $G$ of all north glues of $P$, that are visible from the east, between the lowest glue of $P$ and $y_1$.
  Now, let $i$ and $j$ be two distinct indices, such that $P_i$'s output side is the north, and $P_i$'s north glue is visible from a horizontal line to the east, and $P_j$'s output side is the south, and visible from a horizontal line to the east.

  Suppose, for the sake of contradiction, that $y_{P_i}<y_{P_j}$, and let $l_i$ and $l_j$ be the visibility rays of $P_i$ and $P_j$ respectively.
  There are two (similar) cases, summarized on Figure~\ref{fig:above}:
  \begin{itemize}
  \item Either $i<j$, in which case $P$ must cross either $l_i$ or $l_j$ after $P_j$ before reaching its last point: indeed, that last point is above both lines, and above $P_{i,i+1,\ldots,j}$, since we considered only glues below $y_1$ and since $P_{|P|}$ is a highest point of $P$.

    However, this means that glues $(P_i,P_{i+1})$ and $(P_j,P_{j+1})$ cannot be both visible from a horizontal ray to the east, which is a contradiction.

  \item Or $i>j$, in which case $P$ must also cross either $l_i$ or $l_j$ after $P_i$, for the same reason, also leading to a contradiction.

  \end{itemize}

  \begin{figure}[ht]
    \centering
    \begin{subfigure}[b]{0.45\textwidth}
      \begin{center}
        \includegraphics{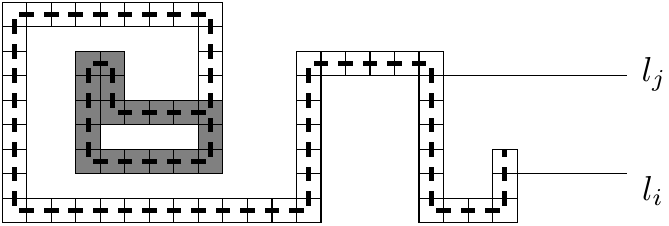}
        \caption{Case $i<j$}
        \label{fig:l1}
      \end{center}
    \end{subfigure}
    \begin{subfigure}[b]{0.45\textwidth}
      \begin{center}
        \includegraphics{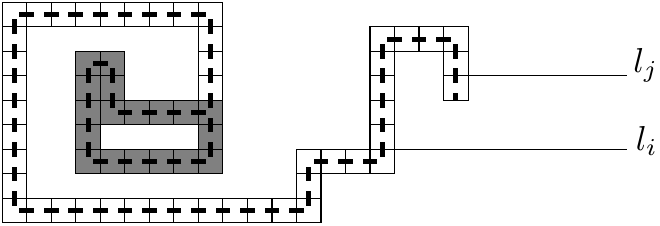}
        \caption{Case $j>i$}
        \label{fig:l2}
      \end{center}
    \end{subfigure}
    \caption{Lemma \ref{lem:watershed0} in action. On both figures, the seed is in gray, the path is uncolored. $P_i$ is the lower tile, and $P_j$ is the higher tile, and the last tile of $P$ is a highest tile. Therefore, in either case, $P_j$ cannot possibly have a visible south output glue, since $P_{[j,|P|]}$ needs to cross one of the visibility rays before reaching its last tile.}
    \label{fig:above}
  \end{figure}

\end{proof}

\begin{lemma}
  \label{lem:order}
  Let $\mathcal{T}=(T,\sigma,1)$ be a tile assembly system such that $\sigma$ is finite and connected, and $P$ be a path producible by $\mathcal{T}$, whose last point is a highest point of $P$.

  If $i$ and $j$ are two integers such that $i<j$, and the north (\resp south) glues of $P_i$ and $P_j$ are both visible from the east, then $y_{P_i}<y_{P_j}$ (\resp $y_{P_i}>y_{P_j}$).

\begin{proof}
  We prove this only for the case where the north glues of $P_i$ and $P_j$ are visible from the east. The proof for their south glues being visible from the east is symmetric.
  Assume, for the sake of contradiction, that $P_j$ is above $P_i$ (see Figure \ref{fig:visibleBefore}).

  We draw a horizontal ray $l_0$ from $P_i$ to the east in $\zs$, and a horizontal ray $l_1$ from $P_j$ to the east in $\zs$. Formally, $l_0$ (\resp $l_1$) is the set of all vertices of the grid graph of $\mathbb{Z}^2$ that are on the same row and to the right of $P_i$ (\resp of $P_j$).

  Now, $P_{i,i+1,\ldots,j}$, $l_0$ and $l_1$ define a cut of $\zs$ into exactly two connected components (indeed, they are disjoint, by construction). Let $\mathcal C$ be the component connected to the points between $l_0$ and $l_1$. Since the north glue of $P_j$ is visible from the east, $P_{j+1}$ is inside $C$. Now, since the last point of $P$ is a highest point of $P$, that last point cannot be inside $\mathcal C$. Therefore, $P$ has to place at least one tile on either $l_0$ or $l_1$, but the only way to do so is by using a vertex to the right of either $P_{i}$ or $P_j$. However, this contradicts our hypothesis that the north glues of both $P_i$ and $P_j$ are visible from a ray to the east.

  \begin{figure}[ht]
    \begin{center}
      \includegraphics[scale=1]{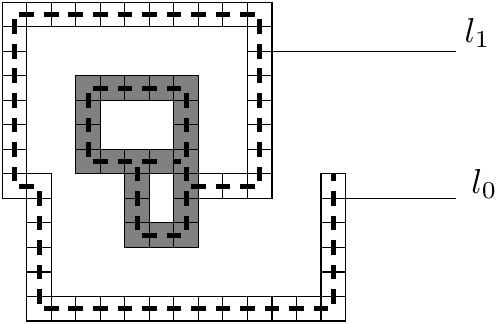}
    \end{center}
    \caption{This figure shows why a path whose last point is its highest has its visible glues built in the same order as the order of their y-coordinates: indeed, if the last tile of $P$ is a highest tile, then $P$ needs to cross at least one of $l_0$ or $l_1$ to reach its last tile, contradicting the visibility of the corresponding glues assumed in Lemma~\ref{lem:order}. On this figure, the seed is in gray, and $P$ is uncolored.}
    \label{fig:visibleBefore}
  \end{figure}
\end{proof}
\end{lemma}

\subsection{Dominating tiles and dominating rays}

A convenient notion that we use is that of \emph{dominating tiles}. A tile $P_i$ of a path $P$ is said to be \emph{dominating for a vector $\vec v\in\mathbb{Z}^2$} if $P$ does not intersect a ray of vector $\vec v$ from $P_{i}$, called the \emph{dominating ray of $P_i$}.

\section{Roadmap of this paper}
\label{sec:roadmap}
\subsection{An algorithm to build ``stake paths'', and break or pump U-turns}
\label{roadmap-algo}
The first step of our proof is an algorithm that tries to pump or block our path $P$. Unfortunately, this does not always succeed, but when it fails, this algorithm builds what we call a \emph{stake path}, which is a path that can ``hold'' a suffix of $P$, and can grow in two different translations.

This algorithm is therefore used for two different duties:
\begin{itemize}
\item Prove that paths containing a special structure called a \emph{nice U-turn} are pumpable or fragile. This is one of our first results, proven in Lemma~\ref{lem:uturn}. A nice U-turn is a prefix of $P$ such that:

  \begin{enumerate}
  \item There are two indices $i<j$ of $P$ whose output side is the north, and whose north glue is visible from the west.
  \item There is an index $k>j$ such that $P_k$ is visible \emph{from the east} (i.e. the other side) on $\sigma\cup P_{[1,k]}$.\label{cond:vis}
  \item $P_k+\vect{P_jP_i}$ is below, or on the same row as, the lowest tile of $P_{[i,k]}$.\label{cond:nobacktrack}
  \end{enumerate}

\item Build stake paths, that will be useful at the end of the proof (in several lemmas of Section~\ref{sec:conc}) to ``move'' the position of conflicts, using the following general argument: try to pump a segment $P_{[u,v]}$ of $P$, held by the stake path (i.e. without growing $P$ up to this segment), and if you fail at pumping, we can prove that there is a conflict between the pumping and a tile of $P$ (because the stake path contains only tiles of $P$, and translations of tiles of $P$).

  Then, grow a translation of the pumping of $P_{[u,v]}$, using a different translation of the stake path, to move the position of the conflict one step forward in the pumping: $P$ will not be able to grow to the position of the conflict from the resulting assembly.
\end{itemize}

This part of our proof is possibly the trickiest, geometrically speaking. To avoid as many ambiguities as possible, we only show invariants on a \emph{program}, described completely in Section~\ref{sec:algo}, and implemented in an online version\footnote{\href{http://users.ics.aalto.fi/meunier/pompe}{http://users.ics.aalto.fi/meunier/pompe}}. Feel free to try and break or pump your own paths!

\subsection{Initial conditions}
\label{roadmap-init}
An essential condition to the success of the next steps of our proof is to find visible north glues, close enough to the seed.
In Section~\ref{sec:initial}, we prove (in Lemma~\ref{lem:init:cond}) that we can find two tiles $P_i$ and $P_j$ of the same type, both with north output side, and north glues visible from the west on $P$, such that $P_i$ and $P_j$ are both on the shortest prefix of $P$ that escapes a rectangle $\rdanger$ of height $2|T|+|\sigma|$ (and width to be defined) vertically centered on $\sigma$.

That lemma basically solves the problem that the visibility of a tile on $P$ and on a prefix $P_{[1,k]}$ of $P$ might be different: indeed, $P_{[k+1,|P|]}$ might ``hide'' tiles visible on $P_{[1,k]}$. However, we show that if this happens, we can find a nice U-turn, and use Lemma~\ref{lem:uturn} to break or pump $P$.

Section~\ref{sec:initial} also shows that we can assume, without loss of generality, that $P$ never goes away from the seed by more than $2|T|$ rows to the south (of course, $P$ can grow arbitrarily far to the north, and arbitrarily far to the east and west).

That Section yields the first bound~$\seedbound=2|T|+|\sigma|$, on the \emph{height} of the ``danger zone'' $\rdanger$.

\subsection{Reset lemma}
\label{roadmap-reset}
We then introduce a notion of \emph{resources} for a path: an adversary trying to build a large non-fragile non-pumpable path, could try to defeat our pumping attempts by drawing a path that comes back to rows close to the seed, i.e. in rectangle $\rdanger$, very often. Then, any segment we try to pump would not even be producible past the first iteration, since that iteration would immediately conflict with these initial parts of the assembly.

In Section~\ref{sec:reset}, we introduce two powerful tools to avoid this:

\begin{enumerate}
\item We first show that if $P$ starts by going far enough to the east or the west before growing $\seedbound$ rows above $\sigma$, then $P$ is pumpable or fragile. This result, shown in Lemma~\ref{lem:band}, uses the Window Movie Lemma from~\cite{Meunier-2014}. This defines the width of rectangle~$\rdanger$.
  \label{arg1}

\item Then, we show in Lemma~\ref{lem:banana}, that there is a height $\sbound$ such that if $P$ goes back to $\rdanger$ after reaching height $\sbound$, $P$ must have a nice U-turn. This proof uses the following argument: if $P$ goes back to $\rdanger$ after height $\sbound$, $P$ must have a U-turn of depth $\sbound$. The problem is, this U-turn might not be nice. However, if it is not nice, this means that another part of $P$ either builds large U-turns to the south (to defeat condition~\ref{cond:nobacktrack} of niceness, as defined above in Section~\ref{roadmap-algo}) or else build large U-turns to the east of our main U-turn (to defeat condition~\ref{cond:vis} of niceness, as defined above in Section~\ref{roadmap-algo}).

  Moreover, we will show that each of these U-turns must get back to $\rdanger$, which can only be done a constant number of time by our argument~\ref{arg1} above.

  We call Lemma~\ref{lem:banana} the \emph{Reset Lemma}, because it shows that $P$ cannot ``access'' its initial part anymore after some height. This is an important lemma in our proof, and we will actually need more ``resets'' later (to show that $P$ cannot access the stake path built by our algorithm).

\end{enumerate}

\subsection{Finding a stake path}
\label{roadstake}
After the reset lemma, we are therefore left with a (quite high) path that does not enter $\rdanger$ anymore.

We then use the last of our ``main weapons'': stake paths. In Section~\ref{sec:algo}, we prove the following invariant on our algorithm: at each step, the algorithm has a ``current'' stake path, which is a path $S$ made of segments of $P$ and translations by $\vect{P_iP_j}$ of segments of $P$, such that both $S+\vect{P_jP_i}$ and $S$ can grow without crossing $P$ (possibly not at the same time, i.e. $S$ and $S+\vect{P_jP_i}$ might conflict).

We also prove that each step of the algorithm advances the index on $P$ reached by the last point of the current stake path $S$ by at least one unit, except possibly in the case that a suffix $P_{[c,|P|]}$ of $P$ can grow without intersecting its translation $P_{[c,|P|]}+\vect{P_jP_i}$.
There are two cases:

\begin{itemize}
\item Either the current stake path built by the algorithm grows above $\sbound$, but stays entirely below row
$\bbound_1$ (to be defined),
in which case we stop the algorithm immediately after building the first such stake path.

In this case, we are ready to proceed to the last step of our proof.

\item Or there is no such step, which means that a suffix $P_{[c,|P|]}$ of $P$ can grow without intersecting its translation $P_{[c,|P|]}+\vect{P_jP_i}$. We call paths where this happens \emph{cage-free paths}, and their case is treated separately by Lemma~\ref{lem:cagefree}, which shows that any cage-free path of a sufficiently large height is pumpable or fragile (this part uses the window movie lemma from~\cite{Meunier-2014}).
\end{itemize}

After this step, we are therefore left with a stake path intersecting $P$ for the last time within the first $\bbound_1$ rows above $\sigma$ (see Figure~\ref{fig:stakezone}), where $\bbound_1$ is another bound defined by Lemma~\ref{lem:cagefree}, also a constant in $|\dom\sigma|$ and $|T|$.

\subsection{Concluding the proof}
\label{sum:concl}
If we have made it to this step, we are left with a stake path $S$ starting inside $\rdanger$ from $P_j$, and reaching a tile $P_c$ of $P$ at some constant height (depending only on $|\dom \sigma|$ and $|T|$), above $\sbound$. Moreover, $S$ does not grow above a constant bound~$\bbound$. We then use the reset lemma another time, and get a bound $\ebound_1$, such that $P$ cannot go back to $S$ after reaching a height $\ebound_1$. More precisely, we define a new danger zone $\rstake$ around $S$, of height $\sbound$ and width defined by the Window Movie Lemma (or in Section~\ref{sec:reset}) and use it as a new seed.

Then, if $P_{[c,|P|]}$ does not intersect $P_{[c,|P|]}+\vect{P_jP_i}$ above $\ebound_1$, we can use the cage-free argument of Section~\ref{roadstake} again. Else, $P_{[c,|P|]}$ and $P_{[c,|P|]}+\vect{P_jP_i}$ intersect above $\ebound_1$: either this intersection is a conflict, in which case we can break $P$ by first growing $P_{[c,|P|]}+\vect{P_jP_i}$, or this intersection is not a conflict, in which case we can start to pump a segment.

In this last case, there are two cases for how the intersection can happen

\begin{enumerate}
\item Either this intersection has the same orientation as $\vect{P_iP_j}$, i.e. the intersection is in such a way that $P_d=P_e+\vect{P_jP_i}$, with $e>d$, in which case we can immediately try to pump it (and move on to the last step of our proof).\label{good}

  If this pumping can grow infinite, we are done. Else, we will show that there must be a conflict between the pumping an another part of the assembly. We then use the stake path $S$ produced by the algorithm to translate the position of conflicts, allowing the pumping attempt to grow one extra iteration and break $P$.

\item Or this intersection is in the opposite orientation, in which case the pumping goes to the south, and could still conflict with parts of the assembly below $\bbound$. However, we will show that if we keep running the algorithm, then after a constant number of pumping attempts, at least one of the pumping attempts will be in case~\ref{good}.
  An important point here, is that these iterations do not change the stake path produced first, but are only used to find new candidates for pumping.
\end{enumerate}

Since we need to find an intersection at each iteration, or else apply the cage-free argument, we are going to iterate the cage-free bound $|\rstake|$ times, getting a sequence $\ebound_1, \ebound_2,\ldots,\ebound_{|\rstake|}$. The final bound is that last one, $\ebound_{|\rstake|}$.

\section{The whole proof in two drawings}

This proof goes by several step, each proving a new constraint on what $P$ can do without being pumpable or fragile. In Sections~\ref{roadmap-algo}, \ref{roadmap-init}, \ref{roadmap-reset}, and~\ref{roadstake}, we prove the constraints shown on Figure~\ref{fig:zone1}.

\begin{figure}[!ht]
\begin{center}
\includegraphics[scale=0.7]{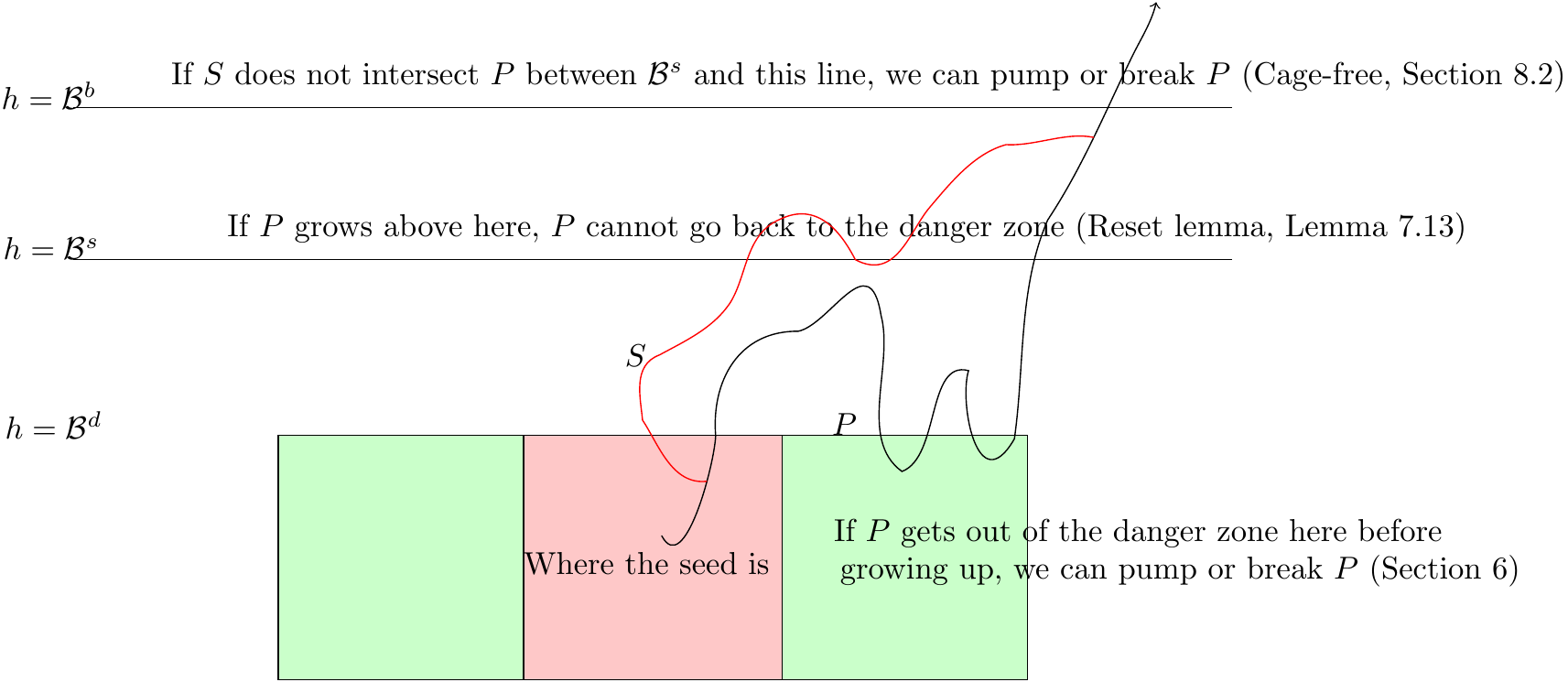}
\end{center}
\caption{First constraints on $P$: if $P$ reaches a certain height, depending only on $|T|$ and $|\sigma|$, $P$ cannot go back to the green zone (called the \emph{danger zone}), and we get a \emph{stake path}, in red, not growing over the top line before intersecting $P$ between the two lines.}
\label{fig:zone1}
\end{figure}

Then, we apply Section~\ref{roadmap-init} and~\ref{roadmap-reset} again, to show that the stake paths (in red) built by the algorithm stay within a bounded ``stake zone'', that the stake cannot escape (the stake zone is in blue on Figure~\ref{fig:stakezone}).

\begin{figure}[!ht]
\begin{center}
\includegraphics[scale=0.7]{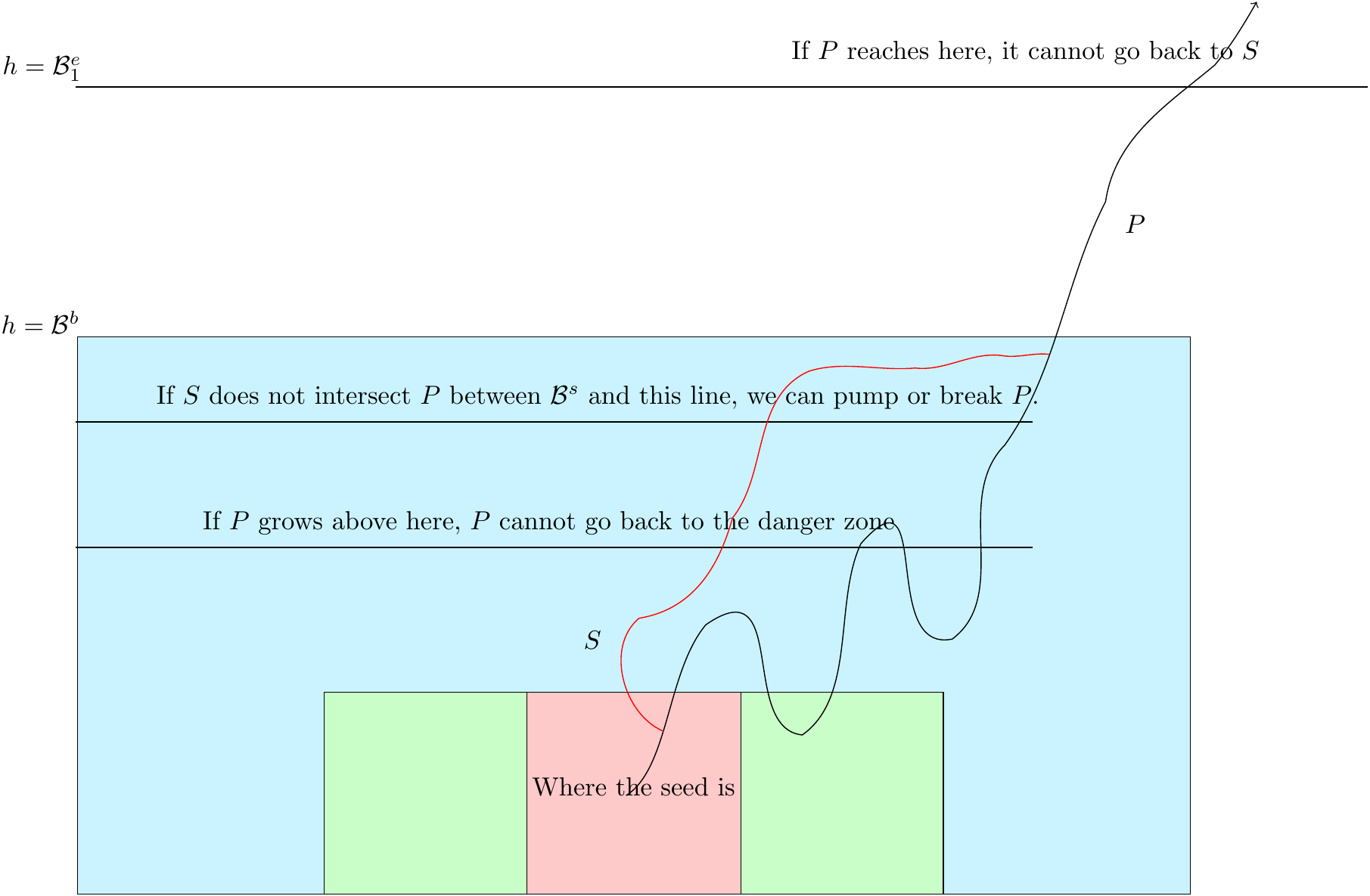}
\end{center}
\caption{Applying Sections~\ref{roadmap-init} and~\ref{roadmap-reset} again, we show that if $P$ grows to a high enough, yet constant, height, then $P$ cannot go back to $S$: if we try to pump $P$ there, that pumping cannot possibly intersect $S$. If $P$ is not pumpable, any attempt to pump $P$ will result in a conflict between the pumping and the end of $P$ (i.e. a part of $P$ above the top line), and we can thus use $S$ to translate conflicts and break $P$. There is a minor complication at this step, in case the pumping is orientated to the south. See Section~\ref{sum:concl} above for a summary, or else Section~\ref{sec:conc}.}
\label{fig:stakezone}
\end{figure}

\clearpage
 \section{Our pumping algorithm}
\label{sec:uturns-lemmas}
In this section, we describe an algorithm to either pump a long enough path $P$, or else construct a ``stake path'', a concatenation of translations of parts of $P$, with special properties that we exploit in the rest of the paper.

\subsection{Our pumping algorithm}
\label{sec:algo}

\paragraph{The goal of this algorithm} is to start from a path $P$ and two visible north glues $P_i$ and $P_j$ of $P$ (with $i<j$) as its input, and to construct a ``stake path'', which is a path $S$, made of parts of $P$ and translations by $\vect{P_iP_j}$ of parts of $P$, and such that $S$ can grow alongside $P$, and so can $S+\vect{P_jP_i}$, without ever crossing $P$.

In the section, we will prove that the pumping always succeeds on paths containing a certain shape called a \emph{nice U-turn}.

The basic functioning of the algorithm is to grow \emph{branches}, which are translation of suffixes of $P$, together with $P$. If a branch shares some positions with $P$, then either the branch and $P$ agree on tile types at these common positions, or they disagree. If they disagree, we can grow the branch first, and prevent $P$ from placing the tile type it wanted to place, which means that $P$ is fragile. If they agree, we can use their agreement to get stronger properties on $P$, allowing us either to try and pump $P$, or to get another branch that can grow longer.

\paragraph{There are several possible conclusions} to this algorithm: if we can (1) pump $P$ infinitely or (2) produce an assembly that prevents $P$ from growing, we are done. However, two other things can happen: if (3) a branch can grow without intersecting $P$, we will prove in Section~\ref{sec:cagefree} that $P$ must be, in some sense, ``straight enough'' to avoid this branch, and therefore must be pumpable (we call this the \emph{cage-free} argument). %

\subsubsection{Variables used by the algorithm}
First, $i$ and $j$ are two fixed parameters, chosen before starting the algorithm, such that $i<j$, the north glues of $P_i$ and $P_j$ are both visible from the west, and of the same type (which happens in particular if $\type(P_i)=\type(P_j)$). At every step, we maintain a \emph{candidate segment}, given by two indices $u$ and $v$ (with no particular order between them), and a \emph{stake path} $S$ such that neither $S$ nor $S+\vect{P_jP_i}$ conflicts with $P$, and they both start and end on the left-hand side of $P$, and never cross $P$.

\subsubsection{The algorithm itself}

Our algorithm has two modes: a \emph{forward} mode, in which we grow translations by $\vect{P_iP_j}$ of segments of $P$, and a \emph{backward} mode, in which we grow translations by $\vect{P_jP_i}$ of segments of $P$.

Intuitively, we keep adding segments (from $P$ and $P+\vect{P_iP_j}$) to $S$, checking every time that the segments do not cross $P$, although they can of course intersect $P$ and stay on the same side (the left-hand side of $P$).

Initially, we start in forward mode with parameters $(i,j,\emptyset)$. Then, we do the following:
\begin{itemize}
\item In forward mode, we first grow $\alpha=\sigma\cup P_{[1,j]}\cup S$, and then try to grow the pumping of $P_{[u,v]}$ (without growing $P_{[j,u]}$). If the pumping of $P_{[u,v]}$ can grow infinite without conflicting with $\alpha$, {\bf we stop the algorithm}: indeed, either $P$ can still grow (which means $P$ is pumpable), or $P$ cannot grow anymore (which means that $P$ is fragile).

  Else, we grow the longest prefix $R$ of $P_{[u,|P|-1]}+\vect{P_iP_j}$ that can grow from $\alpha$, without turning right from $P$ (since $S$ always ends at $P_v$, $R$ starts at $P_v$). In other words, we stop $R$ at the first intersection where $P$ turns left from $P_{[u,|P|-1]}+\vect{P_iP_j}$.

  If $P_{[u,|P|-1]}+\vect{P_iP_j}$ does not intersect $P$, or does not turn right from $P$, {\bf we stop the algorithm}. Else, there are two integers $u'$ and $v'$ such that $P_{u'}=P_{v'}+\vect{P_iP_j}$, in which case we go to backwards mode with parameters $(u',v',S\cup R)$.

\item In backwards mode, we grow $\alpha=\sigma\cup P_{[1,i]}\cup (S+\vect{P_jP_i})$, and then the longest prefix $R$ of $P_{[v,|P|-1]}+\vect{P_jP_i}$ that can grow without turning right from $P$. In other words, we stop $R$ at the first intersection with $P$ where $P$ turns left from $P_{[v,|P|-1]}+\vect{P_jP_i}$.

  If $P_{[v,|P|-1]}+\vect{P_jP_i}$ does not intersect $P$, or does not turn right from $P$, {\bf we stop the algorithm}. Else, there are two integers $u'$ and $v'$ such that $P_{u'}=P_{v'}+\vect{P_jP_i}$, in which case we go to forward mode with parameters $(u',v',S\cup R)$.
\end{itemize}

\subsubsection{Remark on indexing}

In this algorithm and its proof in Section~\ref{algo-proof}, we use the same indices for $P$ and the branches, because branches are suffixes of $P$. For $S$, we use a different indexing, starting from $1$ like any other sequence.

This can also be observed in our implementation, where branches are represented by the index on $P$ of their first tile. In our implementation, $S$ is represented by its domain, i.e. a set of positions (neither the tile types nor the origin of components of $S$ are specified).

\subsection{Paths with nice U-turns are fragile or pumpable}
\label{algo-proof}

In this section, we prove our first breaking/pumping result: in the case that $P$ contains a \emph{nice U-turn}, we can pump or break it. Our proof works by proving invariants on the algorithm in Lemma~\ref{lem:ih}, and then exploiting them in Lemma~\ref{lem:uturn}.

\begin{lemma}[Induction hypothesis]
  \label{lem:ih}
  Let $P$ be a path assembly and $i<j$ two integers such that the output sides of $P_i$ and $P_j$ are the north, and are visible from the west.
  At any step of the algorithm, the following is the case:
  \begin{enumerate}
  \item $S$ is only made of two kinds of parts: segments of $P$, and translations by $\vect{P_iP_j}$ of segments of $P$.\label{inv:twokinds}
  \item $S$ does not cross the visibility ray of $P_i$, and does not cross the visibility of $P_j$ (this implies that $S+\vect{P_jP_i}$ does not cross the visibility ray of $P_i$, although $S+\vect{P_jP_i}$ might cross the visibility ray of $P_j$).\label{inv:vis}
  \item Neither $S$ nor $S+\vect{P_jP_i}$ cross $P$, although both may intersect $P$.
    Both $S$ and $S+\vect{P_jP_i}$ start and end on the left-hand side of $P$.\label{inv:notrespassing}

  \end{enumerate}
\end{lemma}
\begin{proof}
  We prove each item:
  \begin{enumerate}
  \item This can be directly checked in the algorithm.
  \item Indeed, we add to $S$ only parts of $P$, and translations by $\vect{P_iP_j}$ of segments of $P$, but only until they intersect $P$ again: hence, $S$ cannot cross $P$. And since both $P_i$ and $P_j$ have visible north glues, $S$ cannot break the visibility of $P_j$.

  \item This can also be checked in the algorithm: indeed, we change mode every time $S$ intersects $P$.
  \end{enumerate}
\end{proof}

Using this induction hypothesis, we proceed to the main result of this section, showing that paths containing \emph{nice U-turns} are pumpable or fragile.

\begin{lemma}
\label{lem:uturn}
  Let $P$ be a path producible by some tile assembly system $\mathcal{T}=(T,\sigma,1)$.

  If $P$ is such that:
  \begin{enumerate}

  \item There are two indices $i<j$ whose output side is the north, and whose north glue is visible from the west.
  \item There is an index $k>j$ such that $P_k$ has its south glue visible \emph{from the east} (i.e. the other side) on $\sigma\cup P_{[1,k]}$.\label{visible}
  \item $P_k+\vect{P_jP_i}$ is below, or on the same row as, the lowest tile of $P_{[i,k]}$.\label{nobacktrack}
  \end{enumerate}

  Then $P$ is pumpable or fragile.

  (we call the conjunction of these conditions a \emph{nice U-turn}. Section~\ref{sec:reset} deals with more general U-turns)

\end{lemma}
\begin{proof}

  First, we can assume that $P_k$ is the last tile of $P$: indeed, if $P_{[1,k]}$ is pumpable or fragile, then so is $P$.

  We first claim that $P_{[i,k]}$ has at least one dominating tile (see Figure~\ref{fig:dom}): indeed, let $h$ be the index of the highest tile of $P_{[i,k]}$. Since $P_{[1,k]}$ intersects $l_h$ only in $P_h$, $P_h$ is a dominating tile ($P$ might have other dominating tiles before $P_h$).

  Moreover, we claim that the first dominating tile $P_d$ of $P$ is on at least one translated branch during the algorithm. Indeed, if $P$ is neither pumpable nor fragile, the algorithm cannot stop before reaching $P_d$, since the only halt case is then the lack of an intersection between a translated branch and $P$.
  But this cannot happen: indeed, let $l_d$ be the dominating ray of $P_d$ (i.e. a ray of vector $\vect{P_iP_j}$ from $P_d$), and let $l_j$ be the visibility ray of $P_j$. Then, let $\mathcal C$ be a cut of $\mathbb{Z}^2$ made of $l_j$, $P_{[j,d]}$ and $l_d$. Since $S$ starts by turning left from $P$, $S$ is on the left-hand side of that cut, and at any step, neither $P_{[u,|P|]}+\vect{P_iP_j}$ (in forward mode) nor $P_{[u,|P|-1]}+\vect{P_jP_i}$ (in backwards mode) can cross $l_j$ nor $l_d$, by visibility of $P_j$ and domination of $P_d$, respectively.

  \begin{figure}[ht]
    \begin{center}
      \includegraphics{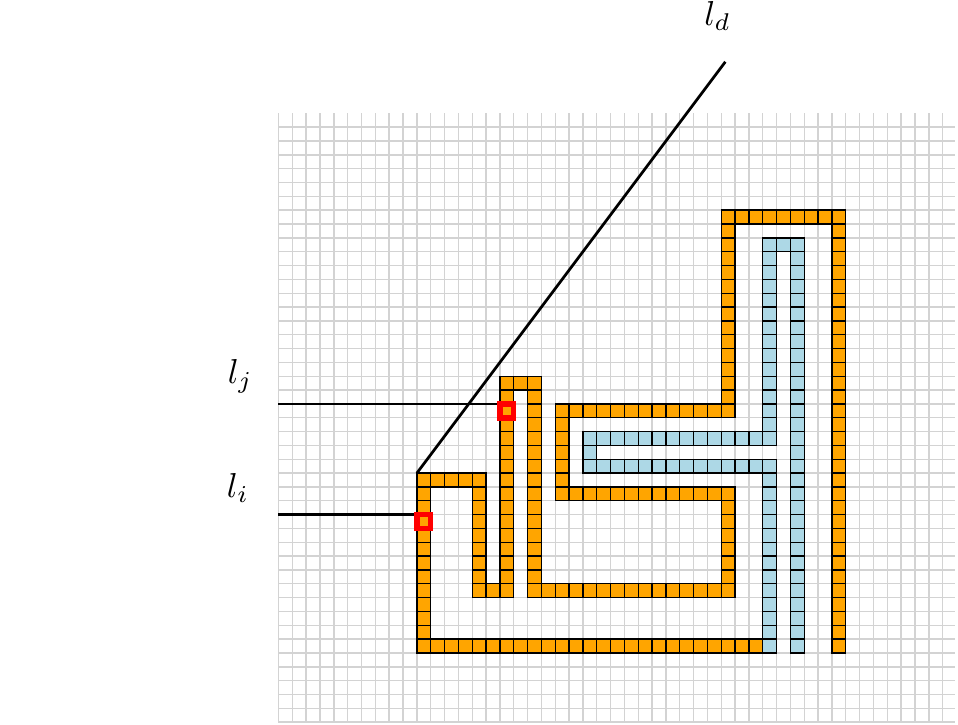}
    \end{center}
    \caption{On this drawing, the seed is in blue and the path in orange. Any path with a nice U-turn has at least one dominating tile, which is not its last one: in this case, the first dominating tile is between $P_i$ and $P_j$, but there are other dominating tiles.}
    \label{fig:dom}
  \end{figure}

  Then, we prove by induction that at each step of the algorithm, if $P_{[u,v]}$ contains a dominating tile, then either $P_{[u,v]}$ is pumpable, or $P$ is fragile, or else $P_{[v+1,|P|]}$ has another dominating tile, and another intersection. The goal is to show that the algorithm cannot halt by not finding an intersection, and therefore, that the only halting cases are if $P$ is fragile or $P$ is pumpable.

  Once the first dominating tile has been reached, each time $P_{[u,v]}$ contains a dominating tile $P_d$, there are two cases:
  \begin{itemize}
  \item Either this happens in backwards mode (see Figure~\ref{fig:backwards}), in which case let $\mathcal C$ be the cut of $\mathbb{Z}^2$ made of the visibility ray of $P_i$, $P_{[i,d]}$ and then the dominating ray $l_d$ of $P_{d}$.

    If $P_{[u,d]}+\vect{P_jP_i}$ cannot grow without crossing $P$, the algorithm moves to backwards mode, with a segment also containing $d$.

    Else, $P_{[u,d]}+\vect{P_jP_i}$ can grow without crossing $P$, and thus $P_{[d,|P|]}+\vect{P_jP_i}$ starts on the left-hand side of this cut, and yet ends on the right-hand side (by hypothesis~\ref{nobacktrack}). However, $P_{[d,|P|]}+\vect{P_JP_i}$ can cross neither $l_i$ (because $P_i$ and $P_j$ are visible from the west), nor $l_d$ (because $P_d$ is dominating).

    Hence, $P_{[d,|P|]}+\vect{P_jP_i}$ must intersect $P_{[i,d]}$, which causes the algorithm to move to forward mode, with a segment containing $P_d$, i.e. containing a dominating tile.

  \begin{figure}[ht]
    \begin{center}
      \includegraphics{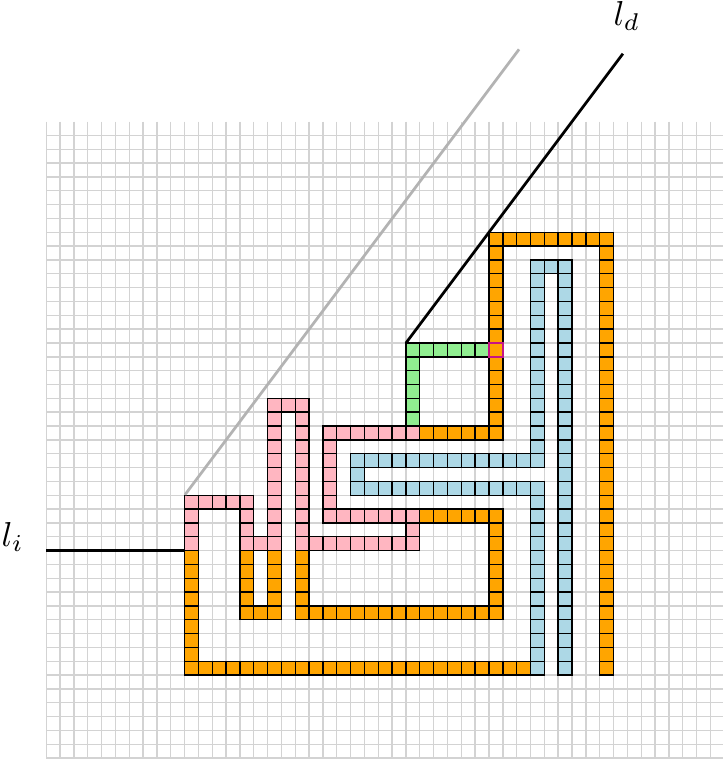}
    \end{center}
    \caption{On this picture, the current stake $S$ is in pink, the seed is in blue, $P$ is in orange. The translation $P_{[v,|P|]}+\vect{P_jP_i}$ (in green), when containing a dominating ray, cannot grow completely without intersecting $P$: indeed, that translation starts by turning left from $P$ (and thus to the left of the cut $\mathcal C$ defined by $l_i$, $P_{[i,d]}$ and $l_d$), but its last point is below all these points, i.e. in the right right of the cut.
    Remark that $P_d$ need not be the first dominating tile: the ray in gray on this picture is from an earlier dominating tile.}
    \label{fig:backwards}
  \end{figure}

  \item Or this happens in forward mode, in which case the algorithm first tries to pump $P_{[u,v]}$ (see Figure~\ref{fig:pumping}).

    First remark that if this fails, the pumping necessarily intersects $P_{[v,|P|-1]}$: indeed, the pumping of $P_{[u,v]}$ grows in a connected component of $\mathbb{Z}^2$ above a cut made of $l_d$, the dominating ray of $P_d$, $P_{[d,k]}$, and $l_k$, the visibility ray of $P_k$ (from $P_k$ to the east).

    And since $P_k+\vect{P_jP_i}$ is at least as low as any tile of $P_{[i,k]}$ (by the nice U-turn hypothesis), no tile of the pumping of $P_{[u,v]}$ can cross $l_k$. Moreover, since $P_d$ is dominating, no translation of $P_{[u,v]}$ by $n\vect{P_iP_j}$ for $n>0$ can cross $l_d$ (although these translations intersect $l_d$ in $P_d+n\vect{P_iP_j}$).

    Therefore, if $P$ is not pumpable, the pumping of $P_{[u,v]}$ must intersect $P_{[v,|P|]}$. This yields a new dominating tile on $P_{[v,|P|]}$: indeed, let $e$ be the index on $P$ of the intersection between $P$ and the pumping of $P_{[u,v]}$, and let $l_e$ be a ray of vector $\vect{P_iP_j}$ from $P_e$. The highest tile $P_f$ of $P_{[v,|P|]}$ on $l_e$ is dominating, since $P_d$ is dominating, and $P_{[d,|P|]}$ does not cross $l_e$ higher than $P_f$.

  \begin{figure}[ht]
    \begin{center}
      \includegraphics{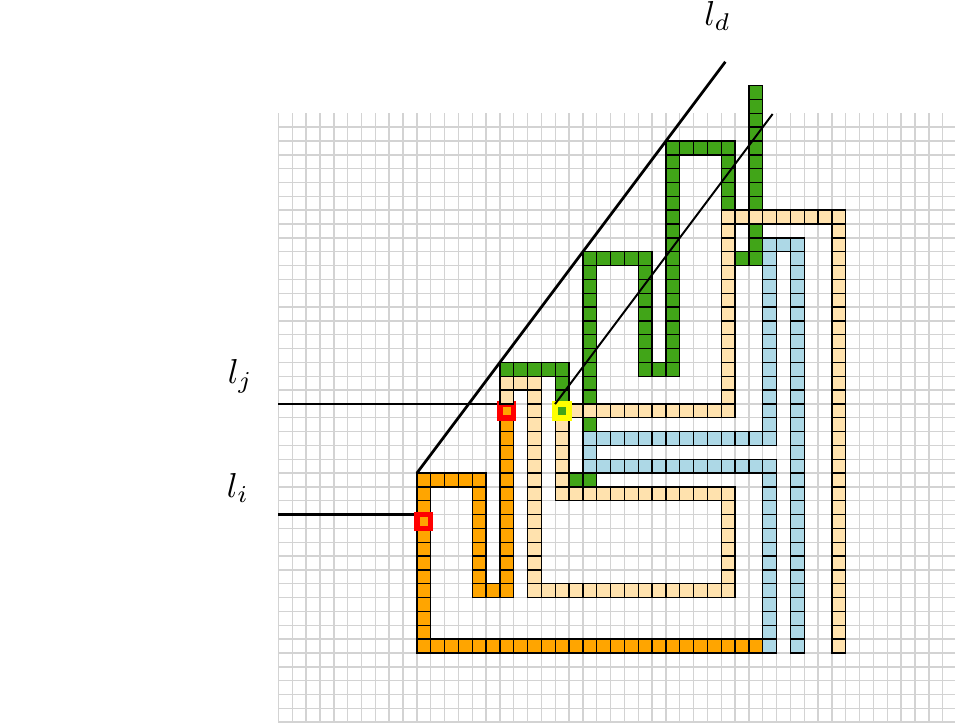}
    \end{center}
    \caption{In forward mode, our algorithm first tries to grow the pumping (in green). If this pumping conflicts with other parts of the assembly (which is the case here), we consider the first time where $P_{[j,|P|-1]}$ turns left from the pumping, and draw a ray $l_k$ of vector $\vect{P_iP_j}$ from that tile. The highest tile on $l_k$ is dominating.}
    \label{fig:pumping}
  \end{figure}

    Moreover, if $P_{[u,|P|]}+\vect{P_iP_j}$ does not intersect $P_{[i,|P|]}$, then $P$ is pumpable: indeed, if $P_{[d,k]}+\vect{P_iP_j}$ can grow completely without intersecting $P_{[i,|P|]}$, then since $P_k$ is visible from the east, further iterations of $P_{[u,v]}$ could also grow from $P_{[d,k]}+\vect{P_iP_j}$ without intersecting anything, since these iterations can cross neither $l_d$ nor $l_k$ (the visibility ray of $P_k$).

    Therefore, if $P$ is not pumpable, the algorithm moves to backward mode, with a segment of $P$ containing at least one dominating tile.

  \end{itemize}

  Finally, the only way the algorithm can stop, in the case that $P$ has a nice U-turn, is by finding a pumpable segment.

\end{proof}

\section{Initial conditions}
\label{sec:initial}
The goal of this section is to find two tiles $P_i$ and $P_j$ of $P$, with $i<j$, whose output sides are both the north side, and such that their north glue is visible from the west. Although this goal may not seem very ambitious, the main complication comes from the details of \emph{visible}: indeed, we want these glues to be visible both on $P$ and on the first prefix of $P$ that grows higher than our first bound $2|T|+|\dom\sigma|$.
The problem is, a glue visible from the west on a prefix of $P$ could become hidden by subsequent parts of $P$.

\begin{lemma}
  \label{lem:init:cond}
  Let $P$ be a path assembly producible by some tile assembly system $\mathcal T=(T,\sigma,1)$, growing at least  $\dbound=2|T|+2$ rows above $\sigma$. At least one of the following is the case:
  \begin{enumerate}
  \item $P$ is fragile
  \item $P$ is pumpable
  \item there are two indices $i<j$ such that, at the same time:\label{third}
    \begin{enumerate}
    \item $P$ has no tile more than $2|T|$ rows below $\sigma$.\label{e}
    \item $P_i$ and $P_j$ are of the same type.\label{a}
    \item Both $P_i$ and $P_j$ are within the first $3|T|+2$ rows around the seed, and $y_{P_i}<y_{P_j}$.\label{c}
    \item $P_{[1,j]}$ is entirely within the first $\dbound=|\dom\sigma|+(3|T|+|\dom\sigma|+5)|T|+1$ rows around the seed.\label{d}
    \item The output sides of $P_i$ and $P_j$ are the north, and are both visible from the west on $P$.\label{b}
    \end{enumerate}
  \end{enumerate}
\end{lemma}
\begin{proof}
  We prove that if case~\ref{third} is not the case, then $P$ has a nice U-turn.

  Let $k$ be the first index on $P$ such that $P_k$ is more than $\dbound$ rows above $\sigma$.
  By Lemma~\ref{lem:watershed0}, since $P_k$ is the last and highest tile of $P_{[1,k]}$, there is a row $y_0$ such that all glues visible on $P_{[1,k]}$ that are above $y_0$ are north glues, and all glues visible on $P_{[1,k]}$ below $y_0$ are south glues.

  Therefore, if $P_{[1,k]}$ has more than $|T|+1$ visible glues below $y_0$, then $P_{[1,k]}$ has a nice U-turn, since $P_k$ is the last and highest tile of $P_{[1,k]}$: this can be checked by considering all three conditions of a nice U-turn, flipped upside down: by the pigeonhole principle and Lemma~\ref{lem:watershed0}, condition (1) of Lemma~\ref{lem:uturn} holds. Then, condition (2) also holds: since $P_k$ is a highest tile of $P_{[1,k]}$, its north glue is visible from both the east and the west. Finally, condition (3) holds because $P_k$ is already a highest tile of $P_{[1,k]}$, hence $P_k+\vect{P_jP_i}$ is only higher (because by Lemma~\ref{lem:order}, the y-coordinates of $P_i$ and $P_j$ are in the reverse order of $i$ and $j$, hence $\vect{P_jP_i}$ is towards the north).

  This proves claim~\ref{e} of our lemma.
  Hence, $y_0$ is at most $|T|+1$ rows above $\sigma$.
  Since $P_k$ is more than $2|T|+2$ rows above $\sigma$, we consider all its visible north glues: by the pigeonhole principle, at least two of them are of the same type, let $i$ and $j$ (with $i<j$) be the indices of two such tiles.

  Up to here, if $P$ is neither pumpable nor fragile, we have therefore already proven claims~\ref{a},~\ref{c},~\ref{e} and~\ref{d} of our statement.

  We now prove claim~\ref{b} of our statement: we already know that $P_i$ and $P_j$ are visible from the west on $P_{[1,k]}$. If $P_i$ or $P_j$ is not visible on $P$, this means that $P_{[k+1,|P|]}$ has at least one glue to the west of the north glue of $P_i$ or $P_j$. But then we claim that this means that $P$ has a nice U-turn, and we can conclude with Lemma~\ref{lem:uturn}. Indeed, we can again check the three conditions of Lemma~\ref{lem:uturn}:
condition (1) of Lemma~\ref{lem:uturn} holds by the pigeonhole principle, and moreover from the size of $\dbound$ we can even choose $P_{i'}$ and $P_{j'}$, two tiles with north glues visible from the east, at least $3|T|+|\dom\sigma+5$ apart. Then, condition (2) of that lemma also holds: indeed, if $P_{[k+1,|P|]}$ has a glue hiding the visibility of $P_i$ or $P_j$, then it must place another visible glue on these rows. Let $k'$ (variable $k$ in Lemma~\ref{lem:uturn}) be the first index of such a tile on $P$. Finally, condition (3) of Lemma~\ref{lem:uturn} is also the case, by claim~\ref{e} of this proof: indeed, $P_{[1,k']}$ has no tiles more than $2|T|+1$ rows below all the tiles of $\sigma$. Therefore, since $P_j$ and $P_i$ are at least $3|T|+|\dom\sigma|+5$ rows apart, $P_{k'}+\vect{P_{j'}P_{i'}}$ is lower than all tiles of $P_{[1,k']}$.

\end{proof}

\section{The Reset Lemma}
\label{sec:reset}
The goal of this section is to define a rectangular zone $\rdanger$ in which $P$ starts, and show that if $P$ enters $\rdanger$ too many times, then $P$ is pumpable or fragile.
That zone is of height $\dbound$ (as defined in Lemma~\ref{lem:init:cond}), and of width exponential in $\dbound$ (we will get more precise in Lemma~\ref{lem:band}).

\subsection{Adapting the Window Movie Lemma to paths}

\begin{lemma}
  \label{lem:adaptwml}
  Let $P$ be a path assembly producible by some tile assembly system $\mathcal T=(T,\sigma,1)$. Let $w$ be a cut-set of $\mathbb{Z}^2$ into two connected components $A$ and $B$, and let $\vec v$ be a non-zero vector of $\mathbb{Z}^2$ such that $w$ and $w+\vec v$ do not have any edge in common.

  Moreover, let $\alpha$ and $\beta$ be the partial assemblies induced by $P$ in $A$ and $B$, respectively, and let $\alpha'$ and $\beta'$ be the partial assemblies induced by $P$ in $A+\vec v$ and $B+\vec v$, respectively.

  If $w$ is such that $\sigma$ is entirely in $A$ and entirely in $A+\vec v$, and the movies of $P$ along $w$ and along $w+\vec v$ are equal up to translation by $\vec v$, then $P$ is pumpable or fragile.

\end{lemma}
\begin{proof}
  The hypotheses of this lemma are already stronger than those of the Window Movie Lemma~\cite{Meunier-2014} (in particular, $w$ and $w+\vec v$ here are required to be disjoint), so we could start by using that lemma. Moreover, since $w$ and $w+\vec v$ do not intersect, we can even use it iteratively, which results in a ``ultimately periodic assembly'' $\gamma$. However, although the ``pumped'' parts all come from $P$, this does not yet mean that $P$ is pumpable, since we have not yet identified any pumpable segment of $P$. And indeed, in some cases, as shown on~Figure~\ref{fig:notapumping}, this might not result in a pumping of $P$.

  \begin{figure}[ht]
    \begin{center}
      \includegraphics{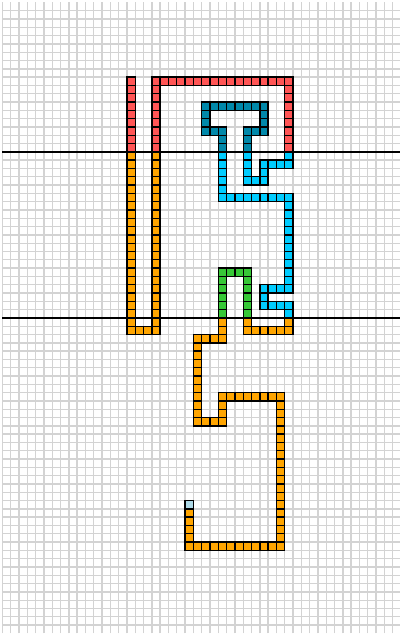}\hfill
      \includegraphics{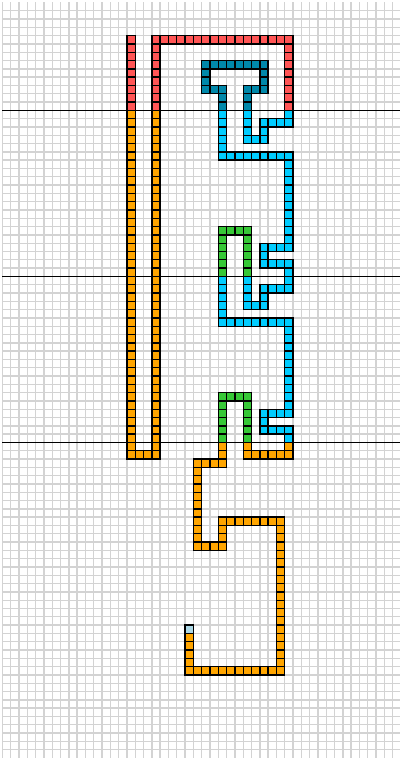}\hfill
      \includegraphics{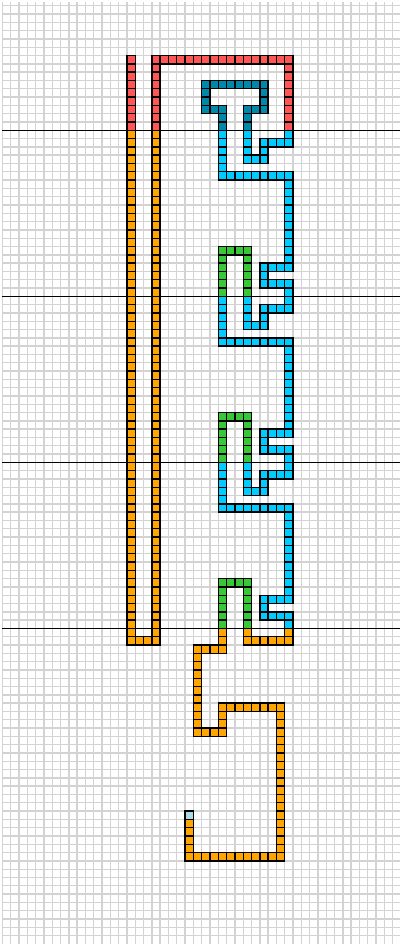}
    \end{center}
    \caption{An example of pumping the assembly induced by a path $P$, in which the resulting assembly is not a pumping of $P$. For this to be the case, the green part, along with both blue parts (light and dark blue) should be repeated between the windows, but here, only the green and light blue are. The seed is just one tile (the lowest blue tile), and the windows are shown as bolder horizontal lines. Other colors are only used to make the parts of $P$ more recognizable.}
    \label{fig:notapumping}
  \end{figure}

  We first identify a candidate segment of $P$ for pumping: let $k$ be the first integer such that the assembly induced by $P_{[1,k]}$ has the same movies on $w$ and $w+\vec v$ (such a prefix of $P$ always exists, since $P$ has the same movies on $w$ and $w+\vec v$), and let then $u$ and $v$ be the indices of the last tiles of $P_{[1,k]}$ that have a glue on $w$ and $w+\vec v$, respectively.

  We will try to pump $P_{[u,v]}$. More precisely let $Q$ be the sequence of points and tile types (not necessarily a path assembly) defined for all $n<u$ by $Q_n=P_n$, and for all $n\geq u$ by $Q_n=P_{(n-u)\mod(v-u)}+\frac{n-u}{v-u}\vect{P_uP_v}$.

  If $Q$ is a path assembly, i.e. if for all $n$, $Q_n$ conflicts neither with $\sigma$ nor with any tile of $Q_{[1,n-1]}$, then $Q$ can grow infinite, and therefore $P$ is pumpable.

  Else, there is a conflict. Since $u$ was chosen to be the last tile of $P$ with a glue on $w$, and $w+\vec v$ is in $B$, $P_{[u,v]}$ is entirely in $B$, and hence, $Q$ cannot possibly conflict with $\sigma\cup Q_{[1,u]}$. The first conflict is therefore between two tiles $Q_s$ and $Q_t$, with $s\in[u,v]$ and $t>v$. However, since $Q_{[u,v]}$ has no tile in $A$ (by choice of $u$), this means that $Q_t$ is in $B'$ (because $Q_t$ is a translation by some $n\vect{P_uP_v}$, with $n>0$, of some tile of $Q_{[u,v]}$). Therefore, this conflict is in $B'$, but $Q_t$ is also in $\gamma$: hence, if we first grow $\gamma$, $Q_s=P_s$ cannot grow anymore, which means that $P$ is fragile.

  \begin{figure}[ht]
    \begin{center}
      \includegraphics{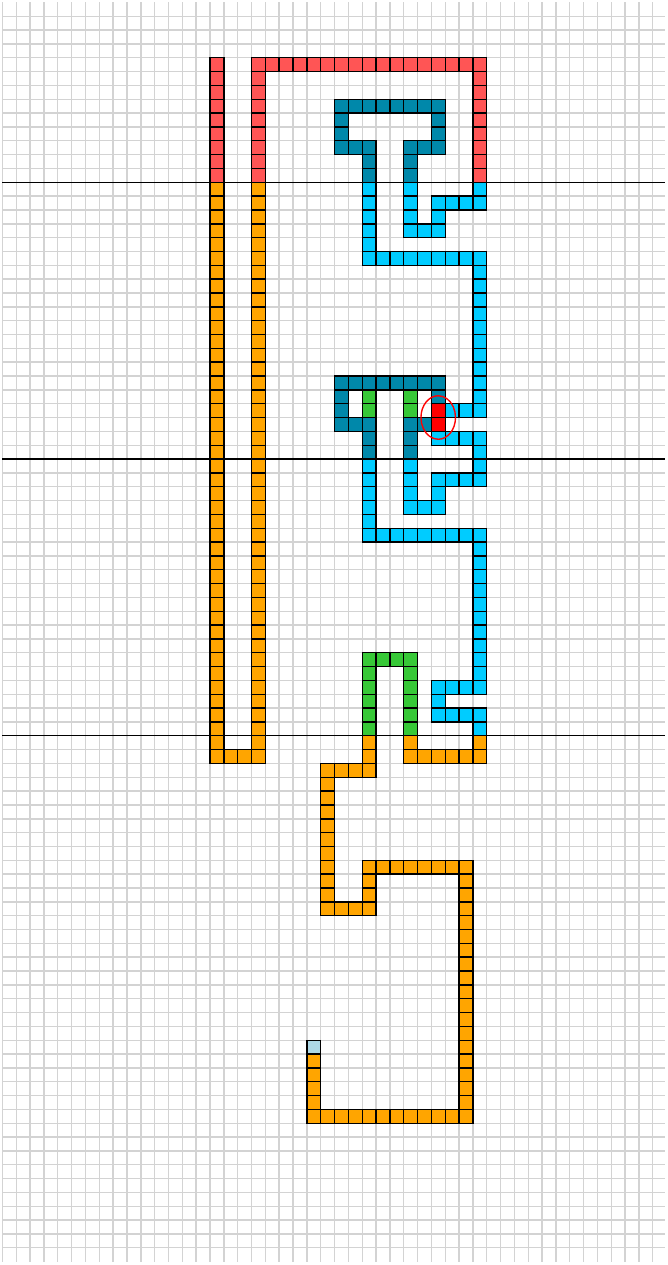}
    \end{center}
    \caption{In the case that $\gamma$ (the assembly ``pumped'' with the window movie lemma) and $Q$ (the pumping of $P$) conflict, as is the case here on the red tiles, this conflict necessarily involves the first iteration $P_{[u,v]}$ of the pumping, which is also in $P$. In this case, our strategy to break $P$ is to simply grow $\gamma$ first: $P$ cannot regrow anymore after that.}
    \label{fig:cuts}
  \end{figure}

\end{proof}

\subsection{Diet paths}

We first show that if a path assembly $P$ grows really long inside a thin stripe of height $\dbound$, then we can pump or block $P$:

\begin{lemma}
  \label{lem:band}
  Let $P$ be a path assembly producible by some tile assembly system $\mathcal T=(T,\sigma,1)$, and let $\rdanger$ be a rectangle of height $\dbound$ and width $2\fband(\dbound)+|\dom\sigma|$, centered around $\sigma$, where for all integer $n>0$, $\fband(n)=((|T|+1)^n)!+1$ is the bound given by the Window Movie Lemma~\cite{Meunier-2014}.

  If $P$ first grows out of $\rdanger$ on the east or west side of $\rdanger$, then $P$ is pumpable or fragile.
\end{lemma}
\begin{proof}
  Let $k$ be the smallest integer such that $P_k$ is outside of $\rdanger$.
  We consider the movies recorded along all possible vertical windows, during the growth of $P_{[1,k]}$.
  Since $\sigma\cup P_{[1,k]}$ stays inside $\rdanger$, the glues on these windows are only on the first $\dbound$ rows around the central row at $y=0$.
  Of course, $P_{[k+1,|P|]}$ may have other glues on these windows, including inside $\rdanger$, but these tiles are not counted in these movies, since we are only trying to pump a \emph{prefix} $P_{[1,k]}$ of $P$.

  Now, assume, without loss of generality, that $P_k$ is to the east of $\rdanger$. Since $P_{[1,k]}$ crosses all the windows from the first one to the east of $\sigma$, up to the easternmost column of $\rdanger$, $P_{[1,k]}$ must have crossed at least $\fband(\dbound)$ different windows, producing at least $\fband(\dbound)$ non-empty movies.
  Hence, at least two of these movies are identical, since there are at most $((|T|+1)^{\dbound})!$ possible movies on a window of height $\dbound$. We can then apply Lemma~\ref{lem:adaptwml}, because the two windows do not intersect.

\end{proof}

\subsection{The Reset Lemma}

We can finally conclude this section by showing that non-diet paths cannot ``remember'' or ``read'' the contents of the danger zone (i.e. initial parts of the assembly, including the seed and early parts of the path) arbitrarily many times, without using a nice U-turn to do so, which will in turn allow us to conclude using Lemma~\ref{lem:uturn}.

This lemma will be useful when we try to pump parts of $P$: indeed, if the pumping of a segment cannot conflict with early parts of the assembly, then the parts with which that pumping conflicts are easy to work with. More precisely, we will build tools (stake paths) to translate the position of conflicts, in order to block $P$.

\begin{lemma}
  \label{lem:banana}
  Let $P$ be a path assembly producible by some tile assembly system $\mathcal T=(T,\sigma,1)$, and let $\rdanger$ be the associated ``danger zone'', of width $w$ and height $h$ (where the exact values of $w$ and $h$ are defined in Lemma~\ref{lem:band}).

  If $P$ grows at least to a height $\sbound=(2|T|)^{w+h}$ rows to the north above $\rdanger$, then either $P$ is pumpable or fragile, or else $P$ cannot have a tile in $\rdanger$ after a tile at height $\sbound$.
\end{lemma}
\begin{proof}
  We prove this by induction on the number of ``free'' positions, i.e. positions not occupied by successive prefixes of $P$, on the border of $\rdanger$. At each step, we look for a nice U-turn in order to conclude that $P$ is pumpable or fragile, and ``consume'' at least one possibility to enter $\rdanger$.

  More precisely, we show, by induction on $n$, that any prefix of $P$ escaping $\rdanger$ $n$ times either grows to at most $(2|T|)^n$ rows to the north of $\rdanger$, or else $P$ has a nice U-turn.
  This clearly holds for $n=0$, since $P$ starts inside $\rdanger$: since $\rdanger$ is a connected component of the grid graph of $\mathbb{Z}^2$, any prefix of $P$ not escaping $\rdanger$ must remain inside $\rdanger$.

  Then, let us assume that for any prefix $P_{[1,k]}$ of $P$ escaping $\rdanger$ exactly $n\geq 0$ times, $P_{[1,k]}$ does not reach height $h_n=(2|T|)^n$. Let then $P_{[1,l]}$ be a prefix of $P$ escaping $\rdanger$ $n+1$ times, i.e. $l>k$, and $P$ enters and escapes $\rdanger$ once after $P_k$.

  If $P_{[k,l]}$ stays below row $2|T|h_n$ before entering $\rdanger$ again, we can conclude that our induction hypothesis still holds for $n+1$.
  Else, $P_{[k,l]}$ grows to row $2|T|h_n$ before entering $\rdanger$ again. We claim that $P_{[1,l]}$ must then have a nice U-turn, between the first and last time $P_{[k,l]}$ grows above row $h_n$: indeed, consider the leftmost and rightmost glues that $P_{[k,l]}$ places on row $h_n$. Since we have proven in Lemma~\ref{lem:init:cond} that $P$ cannot grow below $\rdanger$, we claim that $P_{[k,l]}$ has a nice U-turn: indeed, we check all three conditions of Lemma~\ref{lem:uturn}:
  \begin{enumerate}
  \item First, we can clearly pick two integers $i<j$, by the pigeonhole principle, such that $P_i$ and $P_j$ have their north glue visible from the west (\resp the east, depending on the orientation of $P_{[k,l]}$ on $h_n$). Moreover, we can choose $P_i$ and $P_j$ to be at most $h_n$ rows apart, by the induction hypothesis.
  \item Then, the rightmost (\resp leftmost) glue of $P_{[k,l]}$ on $h_n$ is visible from the east (\resp the west), by definition of visibility, and because $P_{[1,k-1]}$ has no tile on $h_n$.
  \item And finally, since $P_i$ and $P_j$ are so many rows apart, $P_l+\vect{P_jP_i}$ is lower than all tiles of $\sigma\cup P$.
  \end{enumerate}
  Therefore, by Lemma~\ref{lem:uturn}, $P$ is pumpable or fragile.
\end{proof}

\section{Finding a stake}
\label{sec:cage}
Up to here, we know (1) how to break or pump paths with nice U-turns and (2) that any path reaching a certain height, and coming back to a region of constant height and width around the seed has a nice U-turn, and is therefore also pumpable or fragile.

The last case that we need to handle is the case of paths that do not go back to that ``dangerous region'' around the seed, after reaching a certain height. Our general strategy is to let the algorithm go until it finds either a pumpable segment, or else a ``stake path''.

The main problem in this plan is that the algorithm may also halt because it does not find intersections between a branch and $P$. We solve this problem in this section, by bounding for all index $u$, the maximal height that $P_{[u,|P|]}+\vec v$ (where $\vec v=\vect{P_jP_i}$ or $\vec v=\vect{P_iP_j}$ depending on the mode of the algorithm) can reach without intersecting $P_{[u,|P|]}$.

More precisely, we show that if such an intersection happens does not happen early enough, then $P$ is pumpable or fragile for other reasons.

\subsection{Cage free path}
\label{sec:cagefree}

\begin{lemma}
\label{lem:cagefree}
  Let $\vec v$ be a vector of $\mathbb{Z}^2$, such that $y_{\vec v}>0$, and let $P$ be a path assembly producible by some tile assembly system $\mathcal T=(T,\sigma,1)$, whose last point is its highest point.

  If there is an integer $n$ such that $P_{[n,|P|-1]}+\vec v$ does not intersect $P$, and moreover $P$ grows at least $N=n+\|\vec v\|_1(|T|^{\|\vec v\|_1^2}!+1)$ rows north of all the tiles of $\sigma$, then $P$ is pumpable.

\end{lemma}
\begin{proof}
  The goal is to use the lack of intersection to show the existence of a cut $\mathcal C$ of $\mathbb{Z}^2$, such that for all but a ``small'' number of translation vectors $\vec u$, $P$ stays within a bounded central region of $\mathcal C+\vec u$.
  Then, we will conclude using the window movie lemma. Of course, in the bigger picture of the whole paper, $\vec v=\vect{P_iP_j}$, with $P_i$ and $P_j$ the two tiles with visible north glues highlighted previously.

  \paragraph{We first define a family of cuts of $\mathbb{Z}^2$, and show that $P_{[n,|P|-1]}$ intersects them only in a small bounded region.}

  Let first $B$ be a shortest path between $P_n$ and $P_n+\vec v$, and for all $i\geq n$, let $C^i$ be the bi-infinite periodic path made of all translations of $B$ by integer multiples of $\vec v$, i.e.
  for all $k\in\mathbb Z$, let $$C_k^i=B_{k \mod |B|}+\lfloor \frac k {|B|}\rfloor \vec v + \vect{P_nP_i}$$

  For all $i\geq n$, $C^i$ is a vertex separator of the grid graph of $\mathbb{Z}^2$, into two connected components, one to the left and one to the right (indeed, remember that $y_{\vec v}>0$). Let $L^i$ and $R^i$ be the connected components on the left-hand side and right-hand side respectively
  (indeed, since we chose $B$ to be a shortest path between $P_n$ and $P_n+\vec v$, $C^i$ is simple, and therefore ${\mathcal C}^i$ is well-defined and cuts the plane into exactly two connected components).
  Then, let ${\mathcal C}^i$ be the cut of $\mathbb{Z}^2$ defined by ${\mathcal C}^i=(L^i,R^i\cup C^i)$. See Figure~\ref{fig:cagecut} for an example.

  \begin{figure}[ht]
    \begin{center}
      \includegraphics{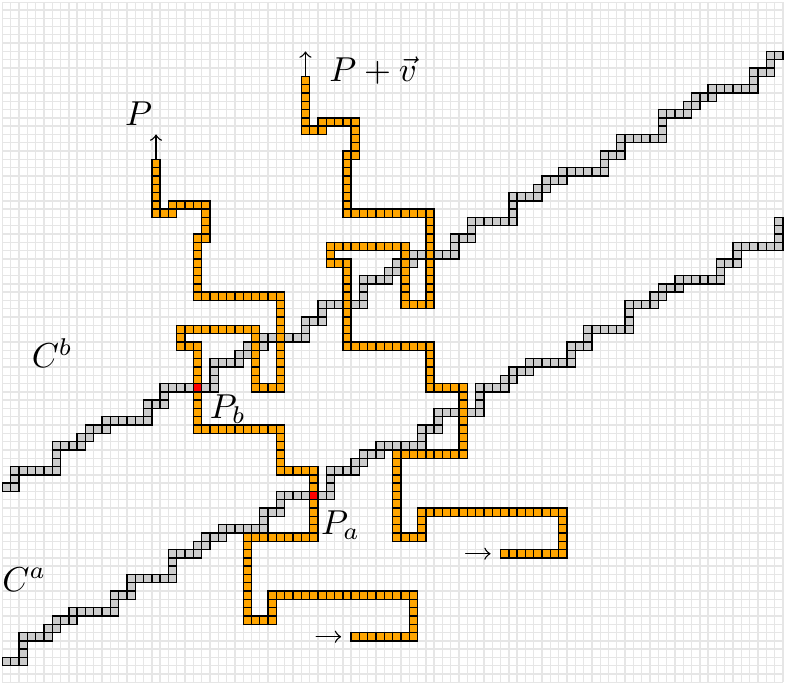}
    \end{center}
    \caption{Two examples of the construction of the vertex separator we use in this proof, for an example $P$ and its translation $P+\vec v$. $P_a$ and $P_b$ are marked in red, and the corresponding vertex separators $C^a$ and $C^b$ are in gray.}
    \label{fig:cagecut}
  \end{figure}

  Until the second step of our proof below (applying the window movie lemma), we consider only intersections between $P_{[n,|P|-1]}$ and $C^i$, and not between the whole path $P$.

  First, let $i_0<i_1<\ldots<i_{m-1}$ be the integers such that $P_{[n,|P|-1]}$ intersects $C^i$ exactly at $P_{i_0}$, $P_{i_1}$, \ldots, $P_{i_{m-1}}$, and $j_0,j_1,\ldots,j_{m-1}$ be the corresponding integers for $C^i$, i.e. for all $k$, $P_{i_k}$ is at the position of $C^i_{j_k}$. Note that $j_k$ might be smaller or larger than $j_{k+1}$ (i.e. the order between $P$ and $C^i$ is not necessarily the same).

  Now, remark that for all $k\geq 0$, if $j_k$ and $j_{k+1}$ are more than one period of $S$ away from each other, then $P_{[n,|P|-1]}$ and $P_{[n,|P|-1]}+\vec{v}$ must intersect: indeed, $P_{[i_k,i_{k+1}]}+\vec v$ starts inside the connected component enclosed by $P_{[i_k,i_{k+1}]}$ and $C^i_{[j_k,j_{k+1}]}$, but its last point $P_{i_{k+1}}$ is outside this connected component. Moreover, by definition of $i_k$ and $i_{k+1}$ as consecutive intersections on $P$, $P_{[i_k,i_{k+1}]}+\vec v$ does not cross $C^i$: therefore, $P_{[i_k,i_{k+1}]}+\vec v$ crosses $P_{[i_k,i_{k+1}]}$. See Figure~\ref{fig:basiccage}.

  \begin{figure}[ht]
    \begin{center}
      \includegraphics{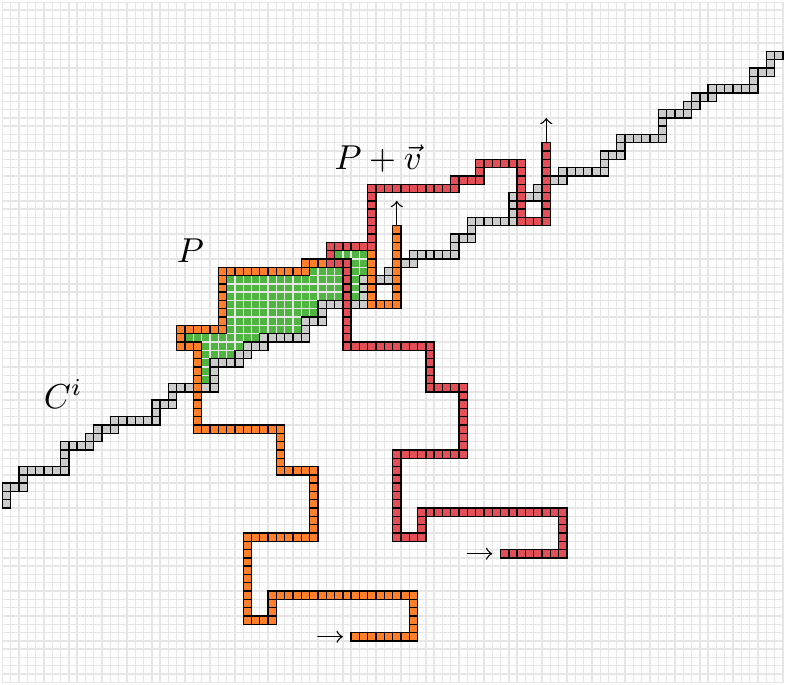}
    \end{center}
    \caption{Whenever two intersections consecutive in $P$, are spaced by more than one period of $S$, they define a connected component (in green) of $\mathbb{Z}^2$, in which $P+\vec v$ starts, but ends outside.}
    \label{fig:basiccage}
  \end{figure}

  To conclude our proof, we only need to apply this observation to two parts of $P_{[n,|P|-1]}$: let $i_l$ be the index of the first intersection between $P_{[n,|P|-1]}$ and $C^i$, in the order of $C^i$ (i.e. $j_l=\min\{j_k | 0\leq k<m\}$).
  \begin{itemize}

  \item First, the intersections between $P_{[i_0,i_l]}+\vec v$ and $C^i$ are either within the first period of $C^i$ from $C^i_{j_0}+\vec v$ and $C^i_{j_l}+\vec v$, or else they are in other periods, which $P_{[i_0,i_l]}$ also intersects (by translation).

    Now, by intersecting $C^i$, $P_{[i_0,i_l]}$ creates ``bumps'' on both sides of $C^i$. One difficulty is, $P$ could alternate between the left and right sides.

    On period $p$ (starting with period 0 between $P_{i_0}+\vec v$ and $P_{i_0}$), let $l_p$ be the minimum between the width of the widest bump on the left hand side, and the width of the widest bump on the right-hand side, where ``width'' of bump $P_{[i_k,i_{k+1}]}$ is $|j_{k+1}-j_k|$.

    If $P_{[i_0,i_l]}+\vec v$ intersects $C^i$ below such a bump $P_{[i_k,i_{k+1}]}$, $P_{[i_0,i_l]}+\vec v$ needs to ``jump over'' that bump, therefore creating another bump, one period below, wider by at least two units.

    Let $P_{[i_a,i_{a+1}]}$ be a segment such that at least $j_a$ or $j_{a+1}$ (or both) is not in the first period of $C^i$ from $C^i_{j_0}+\vec v$. As remarked above $j_a$ and $j_{a+1}$ cannot be more than one period apart from each other (i.e. $|j_{a+1}-j_a|<\|\vec v\|_1$).

    Therefore, the widest bump on each side of $C^i$ keeps increasing by at least two units, on at least one side by period. Therefore, since $C^i$ has two sides, and by our remark above that bumps cannot be wider than $\|v\|_1$, $P_{[i_0,i_l]}$ can only intersect $\|v\|_1$ consecutive periods of $C^i$, i.e. at most $\|\vec v\|_1^2$ different positions of $C^i$.

  \item Then, we can apply the exact same argument to $P_{[i_l,i_{m-1}]}$, with the roles of $P$ and $P+\vec v$ switched.
    Since $P_{i_l}$ is the first intersection (in the order of $C^i$) between $P$ and $C^i$, the maximal span on $C^i$ of the intersections with $P_{[n,|P|-1]}$ is at most $\|\vec v\|_1^2$.
  \end{itemize}

  \paragraph{We then apply the window movie lemma to our family of cuts.}

  This part of the proof is relatively straightforward, although there are two major problems that need to be solved before using the window movie lemma directly (or rather, its adaptation to paths proven in Lemma~\ref{lem:adaptwml}): first, we have too weak hypotheses, since we have only considered intersections of $P_{[n,|P|-1]}$ and $C^i$ (for all $i$). Other parts of the assembly, i.e. $\sigma\cup P_{[0,n-1]}$, could also intersect $C^i$ and change the movies.
  Then, the other problem is, even if we find two different cuts ${\mathcal C}^a$, ${\mathcal C}^b$ with the same movie, these cuts might be close enough that they intersect (i.e. the cuts themselves share common edges), and even though the window movie lemma can be applied to interchange the contents of connected components, this does not necessarily allow us to iterate this operation and pump $P$.

  We first solve the second problem by considering cuts that are far enough from each other, by dividing the plane into ``stripes'' wide enough to contain a cut completely. This means that we consider a new cut every $y_{\vec v}$ rows.

  Then, the first problem can be solved with these cuts: indeed, since $\sigma\cup P_{\sigma\cup [0,n-1]}$ has $|\dom\sigma|+n$ tiles, at most $|\dom\sigma|+n$ cuts are concerned by this problem\footnote{A more subtle argument would take the spacing between stripes into account and get a smaller upper bound.}.

  Then finally, if we see $(|T|^{\|\vec v\|_1^2})! +1$ different windows, we can apply the window movie lemma. In order to see this many windows that $\sigma\cup P_{[0,n-1]}$ does not intersect, we need $P$ to grow to a height at least
  $|\dom\sigma+n|+y_{\vec v}((|T|^{\|\vec v\|_1^2})!+1)$, hence our claim.

\end{proof}

The formulation of Lemma~\ref{lem:cagefree} is very general (the proof technique works for any path in the plane not intersected by its translation), and might be a little obscure for the paths in this paper. Here is a corollary that puts it back to the context:

\begin{corollary}
\label{cor:cagefree}
Let $P$ be a path assembly whose last tile is a highest tile, and is at least $\sbound+\|\vect{P_iP_j}\|_1(|T|^{\|\vect{P_iP_j}\|_1^2}!+1)$ rows above all the tiles of $\sigma$, and such that $P_{[\sbound,|P|]}$ does not intersect $P_{[\sbound,|P|]}+\vect{P_jP_i}$ or $P_{[\sbound,|P|]}+\vect{P_iP_j}$, then $P$ is pumpable or fragile.

Let $\sbound+\|\vect{P_iP_j}\|_1(|T|^{\|\vect{P_iP_j}\|_1^2}!+1)$, which is also a constant in $|T|$ and $|\dom\sigma|$ by Lemma~\ref{lem:band}: indeed, $\|\vect{P_iP_j}\|_1$ is upper-bounded (using the window movie lemma, so that bound, although constant, is quite large in $|T|$ and $|\dom\sigma|$).
\end{corollary}

\subsection{Conclusion of this section}

We can finally conclude that the algorithm finds a suitable stake path, or $P$ is pumpable or fragile (or both).

\begin{lemma}
\label{lem:stake:exist}
Let $P$ be a path assembly producible by some tile assembly system $\mathcal T=(T,\sigma,1)$, and let $i<j$ be the indices of two tiles with visible north glues, as constructed in Section~\ref{sec:initial}.

If the last tile of $P$ is a highest tile of $P$, and is at least $\sbound$ rows above all the tiles of $\sigma$, then at least one of the following is the case:

\begin{enumerate}
\item The algorithm reaches a step with parameters $(u,v,S)$ such that $P_u$ is at least $\sbound$ rows above all the tiles of $\sigma$. Moreover, all the following claims hold:\label{reaches}
\begin{enumerate}
\item $S$ is entirely within the first $\cbound$ rows from the seed.\label{enti}
\item If the algorithm is in forward mode at that step, $S\cup (P_{[u,|P|]}+\vect{P_iP_j})$ is producible from $\sigma\cup P_{[1,j]}$.\label{forw}
\item If the algorithm is in backwards mode at that step, $(S\cup P_{[u,|P|]})+\vect{P_jP_i}$ is producible from $\sigma\cup P_{[1,i]}$.\label{back}
\end{enumerate}
\item $P$ is pumpable\label{pump}
\item $P$ is fragile\label{fra}
\end{enumerate}

Remark that the statements of cases~\ref{forw} and~\ref{back} do not imply that the branches $P_{[u,|P|]}+\vec v$ (where $\vec v=\vect{P_iP_j}$ or $\vec v=\vect{P_jP_i}$ depending on the mode) intersect or don't intersect $P_{[j,|P|]}$.

\end{lemma}

\begin{proof}
  This claim is actually a summary of all previous lemmas:
  First, if the algorithm stops before reaching the claimed step $(u,v,S)$ with $P_u$ at least $\cbound$ rows above all the tiles of $\sigma$, this can only be in two cases (this can be directly read from the algorithm):
  \begin{itemize}
  \item We have successfully blocked or pumped $P$ by first growing the pumping of some segment of $P$.
  \item The algorithm has stopped before that step because it found no intersection, either in forward or backwards mode. We can apply Lemma~\ref{lem:cagefree} and conclude that $P$ is pumpable or fragile.
  \end{itemize}
  This shows that our three cases~\ref{reaches},~\ref{pump} and~\ref{fra} cover all possibilities.
  We now prove our claim for case~\ref{reaches}:
  First, if $S$ grows above $\cbound$ before intersecting the first tile of $P$ above $\sbound$, we can apply Lemma~\ref{lem:cagefree}, since this means that $S$ contains a high segment of $P+\vect{P_iP_j}$ that does not intersect $P$, or a high segment of $P$ that does not intersect $P+\vect{P_jP_i}$.
  This shows if our claim~\ref{enti} is not the case, then $P$ is pumpable or fragile.

  Moreover, the two other claims (claims~\ref{forw} and~\ref{back}) come from Lemmas~\ref{lem:ih} and~\ref{lem:banana}: indeed, we know from Lemma~\ref{lem:ih} that $S$ and $S+\vect{P_jP_i}$ can grow from $\sigma\cup P_{[1,j]}$ and $\sigma\cup P_{[1,i]}$, respectively, and reach $P_v$ and $P_u$, respectively.
  Then, by Lemma~\ref{lem:banana}, we know that neither $P_{[u,|P|]}$ nor $P_{[u,|P|]}+\vect{P_jP_i}$ have tiles in the initial rectangle $\rdanger$, where $\sigma\cup P_{[1,j]}$ is. This concludes our proof, since the claimed translated suffixes of $P$ cannot possibly conflict with any part of the assembly, be it in forward or backwards mode.

\end{proof}

\section{Conclusion of the proof}
\label{sec:conc}

At this stage of the paper, we already have enough weaponry to either block a high enough path assembly $P$, or else pump a segment $P_{[u,v]}$ of $P$ that is such that $\vect{P_uP_v}=\vect{P_iP_j}$, where $i<j$ are the indices of the two tiles with visible north glues given as input to the algorithm. Intuitively, the proof of this claim follows these steps:

\begin{itemize}
\item Using Lemma~\ref{lem:stake:exist}, halt the algorithm at a step where parameters $(u,v,S)$ satisfy the hypotheses of that lemma.
\item Use the algorithm a bit more until finding another ``candidate segment'' $P_{[u',v']}$ for pumping, i.e. another intersection between $P_{[u,|P|]}$ and $P_{[u,|P|]}+\vect{P_jP_i}$.
\item Try to grow the pumping of $P_{[u',v']}$ from $\sigma\cup P_{[1,i]}\cup (S+\vect{P_jP_i})\cup P_{[u,u']}$.
\begin{itemize}
\item If this succeeds, we are done: $P$ is either pumpable if $P$ can regrow after the infinite pumping has grown, or else $P$ is fragile if $P$ cannot regrow.
\item If this fails, there is a conflict between the pumping of $P_{[u',v']}$ and some part of the assembly.
  If moreover $\vect{P_{u'}P_{v'}}=\vect{P_iP_j}$ (this is \emph{not} a general hypothesis, which is why we need this section), then this conflict cannot possibly be with $\sigma\cup P_{[1,j]}\cup S$, which are all to the south of all tiles of $P_{[u',v']}$, and $\vect{P_{u'}P_{v'}}$ is towards the north.

  Therefore, that conflict is between $P_{[u,u']}$ and the pumping. Therefore, by growing instead $S\cup (P_{[u,u']}+\vect{P_iP_j}$, or in other words the translation by $\vect{P_iP_j}$ of all the parts of the assembly above $\sigma\cup P_{[1,i]}$, we can block $P$, by growing the tile placed by the previous iteration of the pumping of $P_{[u',v']}$ at the position of the conflict, instead of the tile placed by $P$ at that position.
\end{itemize}
\end{itemize}

The previous sections have indeed made this plan relatively simple. However, the only other case, where $\vect{P_{u'}P_{v'}}=\vect{P_jP_i}$ is more cumbersome.

\begin{lemma}
  \label{lem:trous}
  Let $P$ be a path producible by some tile assembly system $\mathcal T=(T,\sigma,1)$. If the pumping of two different segments $P_{[a,b]}$ and $P_{[c,d]}$ of $P$ intersect, and are such that $b<c$ and $\vect{P_aP_b}=\vect{P_cP_d}=\vect{P_jP_i}$ (i.e. both pumpings are towards the south), then the algorithm in Section~\ref{sec:algo} also tries to pump a third segment $P_{[e,f]}$ such that $\vect{P_eP_f}=\vect{P_iP_j}$ (i.e in the other direction, towards the north).
\end{lemma}
\begin{proof}
  This is yet another application of dominating tiles: let $P_g$ be the dominating tile of $P_{[b,c]}$

  Immediately after trying to pump $P_{[a,b]}$, the algorithm will try to grow parts of $P_{[b,g]}+\vect{P_iP_j}$ and $P_{[b,g]}+\vect{P_jP_i}$, depending on the mode:
  \begin{enumerate}
  \item In backwards mode,
    let $\mathcal C$ be the cut of $\mathbb{Z}^2$ defined by $l_i$, the visibility ray of $P_i$, $S$, the current stake path at the time of growing a translation of $P_{[b,c]}$, $P_{[b,g]}$ and then $l_g$, a ray of vector $\vect{P_iP_j}$ from $P_g$.

    The parts of $P_{[g,c]}+\vect{P_jP_i}$ grown by the algorithm start from the left-hand side of $\mathcal C$ (by claim~\ref{inv:notrespassing} of Lemma~\ref{lem:ih}), but end on the right-hand side of $\mathcal C$ (at $P_c$, which is on the right-hand of $\mathcal C$ side by definition of $l_g$).
    Hence, they need to cross $\mathcal C$, and reach $P_g$, before the algorithm can first try to pump $P_{[c,d]}$. The only way to do so (by visibility of $P_i$ and the definition of dominating tiles) is if $P_{[g,c]}+\vect{P_jP_i}$ intersects $P_{[b,g]}$, which means that an intersection will be found by the algorithm, that allows us to pump along $\vect{P_iP_j}$. See Figure~\ref{fig:trous}.\label{for}

    \begin{figure}[ht]
      \begin{center}
        \includegraphics[valign=t]{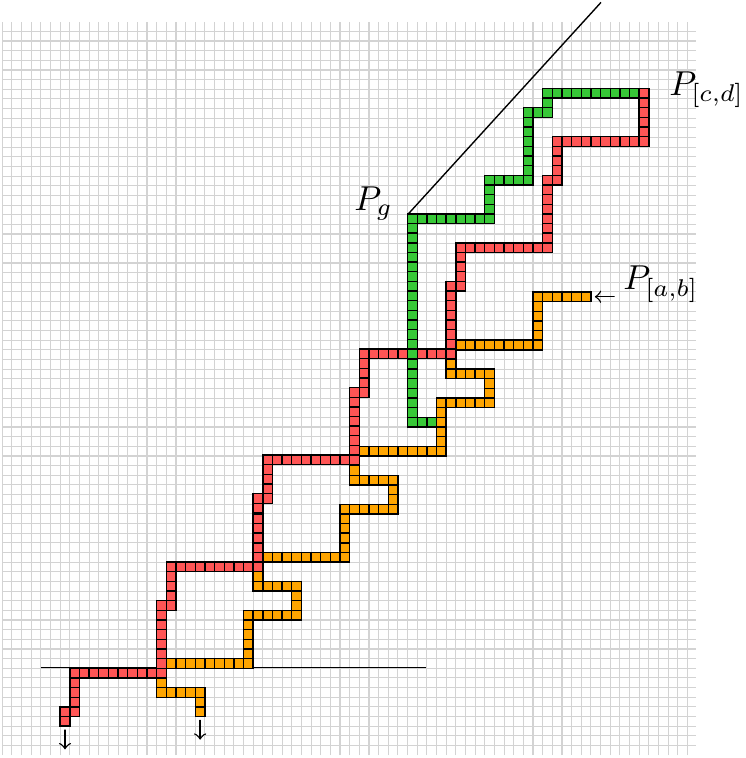}\hfill
        \includegraphics[valign=t]{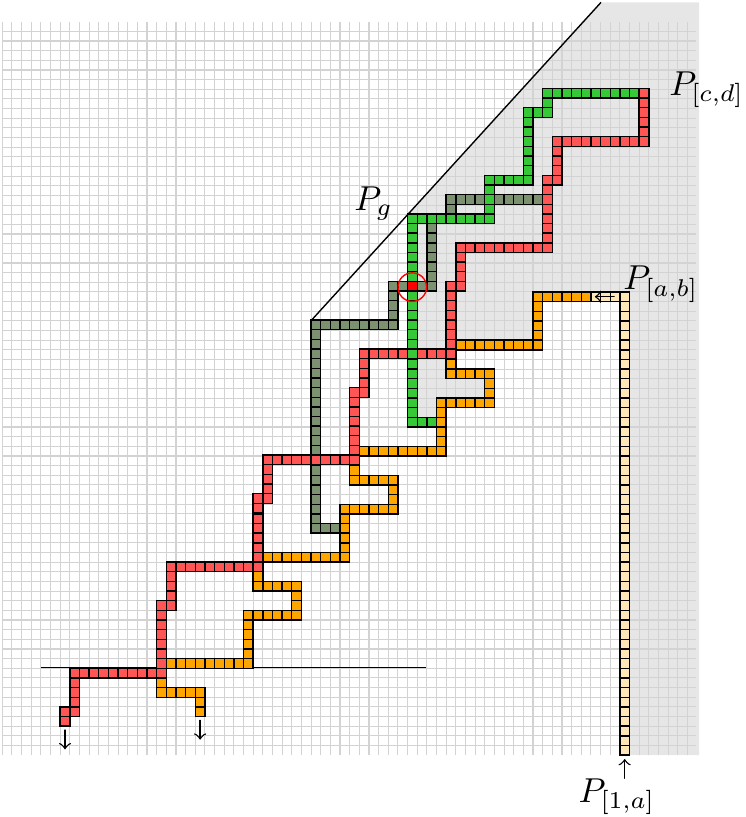}
      \end{center}
      \caption{If two pumpings of segments of $P$ (here, the first pumping, of $P_{[a,b]}$ is in dark orange, and the second pumping, of $P_{[c,d]}$, is in red) intersect, then the segment between them (in green on the drawing on the left) yields another intersection, in the ``correct'' direction, i.e. $\vect{P_iP_j}$ (as opposed to $\vect{P_jP_i}$). This is because the translation $P_g+\vect{P_jP_i}$ of the dominating tile $P_g$ of $P_{[b,c]}$ is outside the connected component enclosed by $l_i$ (the visibily ray of $P_i$), $P_{[i,g]}$ and $l_g$ (the dominating ray of $P_g$), but $P_c$ is inside that component, hence $P_{[g,c]}+\vect{P_jP_i}$ has to cross $P_{[b,c]}$.}
      \label{fig:trous}
    \end{figure}

  \item In forward mode, an intersection will also be found before reaching $P_g$, since $P_g+\vect{P_iP_j}$ (which is on the translation grown by the algorithm in forward mode) is on the right-hand side of $\mathcal C$, but $S$ starts on the left-hand side (by claim~\ref{inv:notrespassing} of Lemma~\ref{lem:ih}). The algorithm will then move on to forward mode, which is handled in case~\ref{for} above.

  \end{enumerate}
\end{proof}

\begin{theorem}
  Let $P$ be a path assembly producible by some tile assembly system $\mathcal T=(T,\sigma,1)$, such that $P$ grows at least $\fbound$ rows above all the tiles of $\sigma$. Then $P$ is pumpable or fragile.
\end{theorem}
\begin{proof}
  This is an assembly of the previous lemmas, along with a formalization of the plan presented in the beginning of this section.

  Let $(u,v,S)$ be the parameters reached by the algorithm with input $i$ and $j$ ($i<j$, the two indices of tiles with visible north glues defined in Section~\ref{sec:initial}), on the first step such that $P_u$ is at least $\sbound$ above all the tiles of $\sigma$.

  We then use the Reset Lemma (Lemma~\ref{lem:banana}) once again to find a first index $P_c$ of $P$, such that $P_{[c,|P|]}$ is completely above $\bbound_1$ (as defined in Corollary~\ref{cor:cagefree}).
  If this is not possible, then by Lemma~\ref{lem:banana}, $P$ is pumpable or fragile.

  After that, we continue running the algorithm until either it halts without finding an intersection (and concluding with Lemma~\ref{lem:cagefree}), or else trying to pump again.
  Each time we try to pump a new segment $P_{[u',v']}$ of $P$, two different cases can occur in the algorithm:

  \begin{enumerate}
  \item The pumping is in the same direction as $\vect{P_iP_j}$, i.e. $\vect{P_{u'}P_{v'}}=\vect{P_iP_j}$.
    In this case, we grow $\sigma\cup P_{[1,i]}\cup (S+\vect{P_jP_i})\cup P_{[u,u']}$, and then the maximal prefix of the pumping of $P_{[u',v']}$ that can grow.

    If that prefix is infinite, we are done: $P$ is fragile or pumpable, depending on whether or not $P_{[i,u]}$ can still grow from the resulting assembly. This is actually found by the algorithm.

    Else, there is a conflict with some part of the assembly. Since we chose $P_{[u',v']}$ to be completely above $S$, and $\vect{P_{u'}P_{v'}}$ is oriented towards the north, the pumping of $P_{[u',v']}$ cannot possibly conflict with $\sigma\cup P_{[1,i]}\cup (S+\vect{P_jP_i})$. Therefore, the only possible conflict is with $P_{[u,u']}$, at some position $P_k$.

    However, we can now translate that large part of the assembly and prevent $P$ from growing $P_k$: indeed, we can grow $\sigma\cup P_{[1,j]}\cup S$, and then $P_{[u,u']}+\vect{P_iP_j}$, and finally the pumping of $P_{[u',v']}$, that can now grow one iteration further, and in particular grow at the position of $P_k$, which means that $P$ is fragile.\label{goodcase}

  \item The pumping is in the other direction, i.e. $\vect{P_{u'}P_{v'}}=\vect{P_jP_i}$.
    There are two cases:
    \begin{enumerate}
    \item Either the pumping of $P_{[u',v']}$ enters $\rstake$, in which case we proceed to the next step of the algorithm. Note, however, that this cannot happen more than $|\rstake|$ times: indeed, by Lemma~\ref{lem:trous}, if this happens at least $|\rstake|$ times, then case~\ref{goodcase} above also happens before, and $P$ is pumpable or fragile.
    \item Or the pumping does not enter $\rstake$, in which case we use a strategy similar to case~\ref{goodcase} above:
      we start by growing $\alpha=\sigma\cup P_{[1,j]}\cup S\cup P_{[v,v']}$, and from there the maximal prefix of the pumping of $P_{[u',v']}$ that can grow. If this prefix is infinite, we are done, by the same argument as before: either $P_{[j,v]}$ can still grow from $\alpha$, which means that $P$ is pumpable, or $P_{[j,v]}$ cannot grow from $\alpha$, which means that $P$ is fragile.

      If this prefix is not infinite, there is a conflict. But since the pumping does not enter $\rstake$, that conflict can only be with $P_{[v,v']}$, i.e. with $P_k$ for some $k\in\{v,v+1,\ldots,v'\}$. But we can grow a different assembly: $\sigma\cup P_{[1,i]}\cup ((S\cup P_{[v,v']})+\vect{P_jP_i})$, and from there the pumping of $P_{[u',v']}$ until the position of the conflict. That pumping will be able to grow one iteration further, and in particular break $P_k$.

    \end{enumerate}

    In total, we need to iterate the bound given by Lemma~\ref{lem:cagefree} a constant number of times (at most $|\rstake|$ times), yielding bounds $\ebound_1$, $\ebound_2$,\ldots, $\ebound_{|\rstake|}$, which still yields a constant (yet very large) bound.
  \end{enumerate}
\end{proof}

\end{document}